\documentclass[11pt,letterpaper]{article}

\pdfoutput=1
\usepackage[margin=2cm]{geometry}
\usepackage{cite}
\usepackage{algorithm}  
\usepackage{algpseudocode}
\usepackage{subfig}
\floatname{algorithm}{Algorithm}

\usepackage{url}
\usepackage{optidef}
\usepackage{lipsum} 
\usepackage{todonotes} 
\newcommand{\mynioptions}{[color=blue!20,inline]}
\newcommand{\mynoptions}{[color=blue!20]}
\newcommand{\myni}{\expandafter\todo\mynioptions}
\newcommand{\myn}{\expandafter\todo\mynoptions}
\usepackage{amssymb} 
\usepackage{mathrsfs} 
\usepackage{amsmath} 
\usepackage{amsthm} 
\usepackage{graphicx}
\usepackage{tikz}
\usepackage{mathtools}
\usepackage[export]{adjustbox}
\usetikzlibrary{shapes,snakes,tikzmark,fit}
\usepackage{multirow}

\makeatletter
\define@key{Gin}{Trim}
            {\let\Gin@viewport@code\Gin@trim\expandafter\Gread@parse@vp#1 \\}
\makeatother
\newcommand{\trimValues}{580 5 550 25}
\newcommand{\trimValuesConfEll}{350 5 350 25}
\newcommand{\trimValuesWithLegend}{210 5 240 25}
\newcommand{\trimValuesNoLegend}{210 5 920 25}

\newtheorem{corr}{Corollary}
\newtheorem{thm}{Theorem}

\newtheorem{lem}{Lemma}
\newtheorem{rem}{Remark}
\newtheorem{assum}{Assumption}
\newtheorem{prop}{Proposition}
\newtheorem{prob}{Problem}

\newtheorem{innercustomthm}{Problem}
\newenvironment{customthm}[1]
  {\renewcommand\theinnercustomthm{#1}\innercustomthm}
  {\endinnercustomthm}

\usepackage{xspace}
\usepackage{amsmath,amssymb,amsfonts}

\newcommand{\bs}[1]{\ensuremath{\boldsymbol{#1}}}

\newcommand{\bp}{\ensuremath{\bs p}\xspace}

\newcommand{\bfs}{\ensuremath{\bs s}\xspace}

\newcommand{\bv}{\ensuremath{\bs v}\xspace}
\newcommand{\bw}{\ensuremath{\bs w}\xspace}
\newcommand{\bx}{\ensuremath{\bs x}\xspace}
\newcommand{\by}{\ensuremath{\bs y}\xspace}

\newcommand{\bW}{\ensuremath{\bs W}\xspace}

\newcommand{\bth}{\ensuremath{\bs \theta}\xspace}

\newcommand{\blam}{\ensuremath{\bs \lambda}\xspace}

\newcommand{\bsig}{\ensuremath{\bs \sigma}\xspace}

\newcommand{\bom}{\ensuremath{\bs \omega}\xspace}

\newcommand{\pci}[1]{\ensuremath{\bs{p}_{\boldsymbol{c}_{i}}}\xspace}

\let\liminf\relax
\let\limsup\relax
\DeclareMathOperator*\liminf{\lim\inf}
\DeclareMathOperator*\limsup{\lim\sup}
\newcommand{\Ball}{\mathrm{Ball}}
\newcommand{\Mybox}{\mathrm{Box}}
\newcommand{\ProjL}{P_{\mathcal{L}}}
\newcommand{\Linner}{\mathcal{L}_\mathrm{inner}}
\newcommand{\Louter}{\mathcal{L}_\mathrm{outer}}
\newcommand{\ymax}{ \overline{y}_\mathrm{max}}
\newcommand{\yshift}{ \overline{y}_\mathrm{shift}}
\newcommand{\Prob}{\mathbb{P}}
\newcommand{\Utau}{\overline{U}_{\tau-1}}
\newcommand{\UN}{\overline{U}_{N-1}}
\newcommand{\bwtau}{\overline{\bw}_{\tau-1}}
\newcommand{\Exp}{\mathbb{E}}
\newcommand{\Projq}{\mathrm{Proj}_{\mathcal{Q}}}

\newcommand{\ProbOvBwtauNoTail}{\Prob_{\overline{\bw}}^{\tau-1}}
\newcommand{\ProbOvBwNNoTail}{\Prob_{\overline{\bw}}^{N-1}}

\newcommand{\hobs}{h_{\mathrm{obs}}}
\newcommand{\xunpert}{\overline{x}_{\mathrm{nodist}}(\tau;\overline{\lambda}_0, \overline{\Lambda}_{\tau-1}, \overline{U}_{\tau-1})}

\newcommand{\xunpertNoTail}{\overline{x}_{\mathrm{nodist}}(\tau;\cdot)}
\newcommand{\ProbbWtau}{\Prob_{\bW}^{\tau}}
\newcommand{\psibWtau}{\psi_{\bW_\tau}}
\newcommand{\psibWtauFull}{\psi_{\bW_\tau}(\overline{y};\tau, \overline{\lambda}_0, \overline{\Lambda}_{\tau-1})}
\newcommand{\psibWtauNoTail}{\psi_{\bW_\tau}(\overline{y};\tau, \cdot)}

\newcommand{\ProbbxLambda}{\Prob_{\bx}^{\tau,\overline{x}_0,\overline{\lambda}_0,\overline{\Lambda}_{\tau-1}, \Utau}}
\newcommand{\psibxzLambda}{\psi_{\bx}(\overline{z};\tau,\overline{x}_0,\overline{\lambda}_0,\overline{\Lambda}_{\tau-1},\Utau)}

\newcommand{\FSRsetDMSP}{\mathrm{FSRset}_S(\tau, \bar{s}_0,\overline{\lambda}_0,\Utau)}

\newcommand{\FSRsetDPV}{\mathrm{FSRset}_P(\tau, \bar{x}_0,\overline{\lambda}_0,\overline{\Lambda}_{\tau-1},\Utau)}
\newcommand{\FSRsetDPVNoTail}{\mathrm{FSRset}_P(\tau;\cdot)}
\newcommand{\Llambdatau}{L_{\lambda,\tau}}
\newcommand{\LlambdatauMinusOne}{L_{\lambda,\tau-1}}
\newcommand{\FSRsetDPVLlambda}{\mathrm{FSRset}_P(\tau, \bar{x}_0,\overline{\lambda}_0,\LlambdatauMinusOne( \overline{q}_{\tau-1}),\Utau)}
\newcommand{\Probbx}{\Prob_{\bx}^{\tau,\overline{x}_0,\overline{\lambda}_0,\LlambdatauMinusOne(\overline{q}_{\tau-1}), \Utau}}
\newcommand{\ProbbxDMSPAS}{\Prob_{\bx}^{\tau,\overline{x}_0,\overline{\lambda}_0,\LlambdatauMinusOne(\overline{q}_{\tau-1}),\Utau}}
\newcommand{\ProbbxNoTail}{\Prob_{\bx}^{\tau}}
\newcommand{\ExpbxNoTail}{\Exp_{\bx}^{\tau}}

\newcommand{\Probby}{\Prob_{\by}}

\newcommand{\psibxzNotail}{\psi_{\bx}(\overline{z};\tau,\cdot)}

\newcommand{\psibxnoarg}{\psi_{\bx}}
\newcommand{\Probbfs}{\Prob_{\bfs}^{\tau, \overline{s}_0, \overline{\lambda}_0, \Utau}}
\newcommand{\ProbbfsNoU}{\Prob_{\bfs}^{\tau, \overline{s}_0, \overline{\lambda}_0}}

\newcommand{\occupDPV}{\phi_{\bx}(\overline{y};\tau,\overline{x}_0,\overline{\lambda}_0,\overline{\Lambda}_{\tau-1},\Utau,h_\mathrm{obs})}
\newcommand{\occupDPVLlambda}{\phi_{\bx}(\overline{y};\tau,\overline{x}_0,\overline{\lambda}_0,\LlambdatauMinusOne( \overline{q}_{\tau-1}),\Utau,h_\mathrm{obs})}
\newcommand{\occupDPVymax}{\phi_{\bx}(\ymax;\tau,\cdot)}
\newcommand{\occupDPVxunpert}{\phi_{\bx}(\xunpertNoTail;\tau,\cdot)}
\newcommand{\occupDPVshift}{\phi_{\bx}(\xunpertNoTail+\yshift;\tau,\cdot)}

\newcommand{\occupDPVinitx}{\phi_{\bx}(\overline{y};\tau,\overline{x}_0,\cdot)}

\newcommand{\occupDPVinitxZero}{\phi_{\bx}\left(\overline{y} - \overline{x}_{\mathrm{nodist}}(\tau);\tau,\overline{0}\right)}

\newcommand{\occupDPVNoTail}{\phi_{\bx}(\overline{y};\tau,\cdot)}
\newcommand{\occupDPVNoTailSeq}{\phi_{\bx}(\overline{y}_i;\tau,\cdot)}
\newcommand{\occupDPVnothing}{\phi_{\bx}}
\newcommand{\occupDPVNoTailbeta}{\phi_{\bx}(\overline{y};\tau,\overline{q}_{\tau-1},\cdot)}
\newcommand{\occupDPVNoTailbd}{\phi_{\bx}( \overline{y}+ b\overline{d};\tau,\cdot)}
\newcommand{\occupDPVNoTailbdc}{\phi_{\bx}( \overline{y}+ b\overline{d}_c;\tau,\cdot)}
\newcommand{\occupDPVNoTailbidc}{\phi_{\bx}( \overline{y}+ b_i\overline{d}_c;\tau,\cdot)}
\newcommand{\occupDMSP}{\phi_{\bfs}(\overline{y};\tau,\overline{s}_0,\overline{\lambda}_0,\Utau,h_\mathrm{obs})}
\newcommand{\occupDMSPNoTail}{\phi_{\bfs}(\overline{y};\tau,\cdot)}
\newcommand{\occupDMSPnothing}{\phi_{\bfs}}
\newcommand{\ASDPV}{\mathrm{PrOccupySet}_P(\alpha;\tau, \overline{x}_0,\overline{\lambda}_0,\overline{\Lambda}_{\tau-1},\Utau,h_\mathrm{obs})}
\newcommand{\ASDPVinitx}{\mathrm{PrOccupySet}_P(\alpha;\tau, \overline{x}_0,\cdot)}
\newcommand{\ASDPVinitxZero}{\mathrm{PrOccupySet}_P(\alpha;\tau, \overline{0},\cdot)}
\newcommand{\ASDPVnothing}{\mathrm{PrOccupySet}_P}
\newcommand{\alphaSqtau}{\alpha_S( \overline{q}_\tau)}
\newcommand{\ASDPValpha}{\mathrm{PrOccupySet}_P(\alphaSqtau;\tau,\cdot,\cdot, \LlambdatauMinusOne(\overline{q}_{\tau-1}),\cdot)}
\newcommand{\ASDPValphaNoTail}{\mathrm{PrOccupySet}_P(\alphaSqtau;\tau,\cdot)}
\newcommand{\ASDPVNoTail}{\mathrm{PrOccupySet}_P(\alpha;\tau,\cdot)}
\newcommand{\ASDPVNoTailunder}{{\mathrm{UnPrOccupySet}_P}(\alpha;\tau,\cdot)}
\newcommand{\ASDPVNoTailover}{{\mathrm{OvPrOccupySet}_P}(\alpha;\tau,\cdot)}
\newcommand{\ASDPVNoTailPrime}{\mathrm{PrOccupySet}_P(\alpha';\tau,\cdot)}

\newcommand{\ASDPVoverMink}{\mathrm{PrOccupySet}_P^+(\alpha;\tau,\cdot)}
\newcommand{\OvASDPVoverMink}{\mathrm{OvPrOccupySet}_P^+(\alpha;\tau,\cdot)}
\newcommand{\ASDMSP}{\mathrm{PrOccupySet}_S(\alpha;\tau, \overline{s}_0,\overline{\lambda}_0,\Utau,h_\mathrm{obs})}
\newcommand{\ASDMSPnothing}{\mathrm{PrOccupySet}_S}
\newcommand{\ASDMSPNoTail}{\mathrm{PrOccupySet}_S(\alpha;\tau, \cdot)}

\newcommand{\ProbbsigTau}{\Prob_{\bsig}^{\tau,q_0}}
\newcommand{\ProbOvbsigTau}{\Prob_{\overline{\bsig}}^{\tau,q_0}}
\newcommand{\ProbOvbsigTauNoTail}{\Prob_{\overline{\bsig}}^{\tau}}
\newcommand{\ProbSigOnly}{\Prob_{\bsig}^t}
\newcommand{\ProbOvbsigTauMOne}{\Prob_{\overline{\bsig}}^{\tau-1,q_0}}

\newcommand{\sigmaAlg}{\mathscr{S}}

\newcommand{\GammaBi}{\Gamma_{b}}
\newcommand{\tausw}{\tau_{\mathrm{s}}}
\newcommand{\QtauProj}{\Gamma_{\overline{q}}(\tau, \mathcal{G}_S)}
\newcommand{\QtauProjObs}{\Gamma_{\overline{q}}\left(\tau, \cup_{q\in \mathcal{Q}} (q, \mathcal{X})\right)}
\newcommand{\QtauMinus}{\mathcal{G}_{\mathcal{Q}^\tau}}
\newcommand{\DPV}{DPV}
\newcommand{\DMSP}{DMSP}

\newcommand{\Occupfunstext}{Probabilistic occupancy functions}
\newcommand{\occupfuntext}{probabilistic occupancy function}

\newcommand{\avoidsettextNoAlpha}{probabilistic occupied sets}
\newcommand{\avoidsetstext}{$\alpha$-probabilistic occupied sets}
\newcommand{\avoidsettext}{$\alpha$-probabilistic occupied set}

\begin{document}
\title{Probabilistic Occupancy Function and Sets Using Forward Stochastic Reachability for Rigid-Body Dynamic Obstacles}
\author{Abraham~P.~Vinod and Meeko~M.~K.~Oishi
\thanks{This material is based upon work supported by
        the National Science Foundation. Vinod and Oishi are supported under
        Grant Number CMMI-1254990 (CAREER, Oishi), CNS-1329878, and IIS-1528047. Any opinions, findings, and conclusions or recommendations expressed in this material are those of the authors and do not necessarily reflect the views of the National Science Foundation.}
\thanks{A. Vinod and M. Oishi (corresponding author) are with Electrical and Computer Engineering, University of New Mexico, Albuquerque, NM 87131 USA; e-mail: \texttt{aby.vinod@gmail.com}, \texttt{oishi@unm.edu}}
}

\maketitle

\begin{abstract}
We present theory and algorithms for the computation of probability-weighted ``keep-out'' sets to assure probabilistically safe navigation in the presence of multiple rigid body obstacles with stochastic dynamics. Our forward stochastic reachability-based approach characterizes the stochasticity of the future obstacle states in a grid-free and recursion-free manner, using Fourier transforms and computational geometry. We consider discrete-time Markovian switched systems with affine parameter-varying stochastic subsystems (DMSP) as the obstacle dynamics, which includes Markov jump affine systems and discrete-time affine parameter-varying stochastic systems (DPV). We define a probabilistic occupancy function, to describe the probability that a given state is occupied by a rigid body obstacle with stochastic dynamics at a given time; keep-out sets are the super-level sets of this occupancy function. We provide sufficient conditions that ensure convexity and compactness of these keep-out sets for DPV obstacle dynamics. We also propose two computationally efficient algorithms to overapproximate the keep-out sets --- a tight polytopic approximation using projections, and an overapproximation using Minkowski sum. For DMSP obstacle dynamics, we compute a union of convex and compact sets that covers the potentially non-convex keep-out set. Numerical simulations show the efficacy of the proposed algorithms for a modified version of the classical unicycle dynamics, modeled as a DMSP. 
\end{abstract}

\begin{keywords} 
    Stochastic reachability; convex optimization; obstacle avoidance; stochastic optimal control; robotic navigation 
\end{keywords}

\section{Introduction}

Stochastic motion planning problems~\cite{thrun_probabilistic_2005, lavalle2006planning} require planning a probabilistically safe path for the navigation of a controllable robot in an environment with multiple stochastically moving rigid body obstacles under bounded control authority.
Most approaches 1) quantify the collision probability, 2) characterize keep-out regions, the set of states that should be avoided to ensure that the collision probability is below a desired threshold, and 3) generate dynamically-feasible trajectories given a set of keep-out regions, to achieve desired properties like minimizing a performance objective, staying within a safe region, and/or reaching a goal.
The first two steps are typically done together using either grid-based approaches~\cite{lavalle2006planning, thrun_probabilistic_2005,elfes1989using, ichter2017real, lambert2008fast, fulgenzi_dynamic_2007, bautin2010inevitable}, chance constraints~\cite{blackmore2011chance,masahiro_ono_iterative_2008,ono2015chance,luders_chance_2010,aoude2013probabilistically,du2011probabilistic}, or reachability~\cite{althoff2009model,SummersHSCC2011, HomChaudhuriACC2017,wu2012guaranteed,malone2017hybrid,chiang2015aggressive}.
The last step may be performed using existing motion planning approaches, such as sampling-based approaches like RRT$^{\ast}$ and PRM$^{\ast}$~\cite{karaman2011sampling} or optimization-based approaches like mixed-integer linear programming \cite{schouwenaars2002safe}, mixed-integer quadratic programming \cite{mellinger2012mixed}, and successive convexification~\cite{mao2017successive}. 
This paper extends our previous work on obstacle avoidance~\cite{HomChaudhuriACC2017} and forward stochastic reachability~\cite{VinodHSCC2017} to formulate a grid-free, recursion-free, and sampling-free approach to quantify the collision probability and characterize the keep-out sets.

Grid-based approaches query an occupancy grid~\cite{elfes1989using,thrun_probabilistic_2005,lavalle2006planning} to assess the collision probability.
The occupancy grid may be updated using probabilistic velocity obstacles~\cite{fulgenzi_dynamic_2007}, probabilistic inevitable collision state~\cite{bautin2010inevitable}, or by sampling~\cite{ichter2017real, lambert2008fast}.
Sampling-based approaches (Monte-Carlo simulations) are popular since they can accommodate rigid body obstacles with nonlinear dynamics.
This versatility comes at a high computational cost when high-quality approximations are desired~\cite{lambert2008fast, calafiore2006scenario}, although importance sampling and the parallelization has improved the computational tractability~\cite{ichter2017real}. 

Chance constraints have been used to plan trajectories for a Gaussian disturbance-perturbed robot navigating an environment with static polytopic obstacles~\cite{blackmore2011chance,masahiro_ono_iterative_2008,ono2015chance}, and extended to obstacles that translate (no rotation) according to a Gaussian process~\cite{luders_chance_2010,aoude2013probabilistically}. 
These approaches replace the probabilistic safety constraints with tighter deterministic constraints that the motion planner must satisfy, and hence are conservative. 
The probabilistic collision avoidance constraint in~\cite{du2011probabilistic} for spherical rigid body robot and the obstacles with Gaussian disturbances was formulated as an integral, and an approximation of the keep-out region was provided.

The third approach is to use backward stochastic reachability via dynamic programming, to compute the inevitable collision states~\cite{malone2017hybrid,chiang2015aggressive}. 
However, these approaches suffer from the curse of dimensionality~\cite{Abate2008,SummersHSCC2011}.
Researchers have proposed particle filters~\cite{blackmore2011chance, lesser_stochastic_2013} and approximate dynamic programming~\cite{kariotoglou2015multi} to improve the computational tractability.
For stochastic rigid body obstacles with a discrete disturbance, a convolution-based formulation was proposed to quantify the collision probability~\cite{HomChaudhuriACC2017}.


The main contributions of this paper are:
\begin{enumerate}
    \item forward stochastic reachability tools to analyze the stochasticity of the state corresponding to \DPV{} and \DMSP{} dynamics at a future time of interest,
    \item definition of the \emph{\occupfuntext{}} via forward stochastic reachability to describe the probability that a rigid body obstacle with \DPV{}/\DMSP{} dynamics occupies a given state at a given time of interest, 
    \item sufficient conditions for the closedness, boundedness, and compactness of the \emph{\avoidsettext{}} ($\alpha\in[0,1]$) and the upper semi-continuity of the \occupfuntext{} for a rigid body obstacle with \DMSP{} dynamics (subsumes \DPV{} dynamics), 
    \item sufficient conditions for the convexity of the \avoidsettext{} and the log-concavity of the \occupfuntext{} for a rigid body obstacle with \DPV{} dynamics, and
    \item two computationally efficient algorithms to compute (tight) approximations of the \avoidsettext{} for a rigid body obstacle with \DPV{} dynamics, that overapproximate the \avoidsettext{} when the dynamics are \DMSP{}.
\end{enumerate}

The paper is organized as follows: Section II provides the problem formation as well as mathematical preliminaries.  Section III describes the system dynamics.  Section IV describes the use of forward stochastic reachability to characterize probabilistic occupancy functions and the $\alpha$-probabilistic occupancy set, and their properties are described in Section V.  Section VI provides two algorithms to compute tight overapproximations, and they are demonstrated on a unicycle example in Section VII.  Conclusions are provided in Section VIII.


\section{Problem statements and preliminaries} 
\label{sec:prem}

We denote a discrete-time time interval by $ \mathbb{N}_{[a,b]}$ for $a,b\in \mathbb{N}$ and $a\leq b$, which inclusively enumerates all integers in between (and including) $a$ and $b$, random variables/vectors with bold case, non-random vectors with an overline, and concatenated random variables/vectors as bold case with overline or with bold uppercase letters.
The indicator function of a non-empty set $ \mathcal{G}$ is denoted by $1_{\mathcal{G}}(\overline{y})$, such that $1_{\mathcal{G}}(\overline{y})=1$ if $\overline{y}\in \mathcal{G}$ and is zero otherwise.
We denote $ I_n\in \mathbb{R}^n$ as the identity matrix, $\overline{0}$ as the zero vector (origin) in $ \mathbb{R}^n$, the Minkowski sum as $\oplus$, the Cartesian product of the set $ \mathcal{G}$ with itself $k\in \mathbb{N}$ times as $ \mathcal{G}^k$, the cardinality of $ \mathcal{G}$ with $ \vert \mathcal{G}\vert$, and the Lebesgue measure of a measurable set $ \mathcal{G}$ by $ \mathrm{m}( \mathcal{G})$.
Given $a,r>0$ and a vector $ \overline{c}\in \mathbb{R}^n$, we denote the $ \overline{c}$-centered Euclidean ball of radius $r$ by $ \Ball( \overline{c}, r) = \{ \overline{y} \in \mathbb{R}^n: {\Vert \overline{y} - \overline{c} \Vert}_2 \leq r\}$ and  the $ \overline{c}$-centered axis-aligned box of side $2a$ as $\Mybox( \overline{c}, a) = \{ \overline{y}\in \mathbb{R}^n: {\Vert \overline{y} - \overline{c} \Vert}_\infty \leq a\}$.

\subsection{Problem statements}

We will first define \DMSP{} dynamics, establish its connections with existing system formulations like DTSHS and Markov jump affine systems, describe \DPV{} dynamics as a special case of \DMSP{} dynamics, and show that a modified version of the classical unicycle can be modeled as a \DMSP{} system.

Next, we will develop the forward stochastic reachability tools~\cite{VinodHSCC2017} to analyze \DPV{}/\DMSP{} dynamics.
\begin{prob}
    Characterize the forward stochastic reachability for \DPV{} dynamics, i.e., construct analytical expressions for
    \begin{enumerate}
        \item the smallest closed set that covers all the reachable states
       (i.e., the forward stochastic reach set).
        \item the probability measure over the forward stochastic reach set (i.e.,
       the forward stochastic reach probability measure) 
    \end{enumerate}
    \label{prob_st:FSR_DPV}
\end{prob}
\begin{prob}
    Characterize the forward stochastic reachability for \DMSP{} dynamics using the forward stochastic reachability for a collection of appropriately defined \DPV{} dynamics.
    \label{prob_st:FSR_DMSP}
\end{prob}
For \DPV{} dynamics, we will also define a forward stochastic reach probability density, and provide sufficient conditions under which the forward stochastic reach set is convex, and the forward stochastic reach probability density and measure are log-concave. 

For rigid body obstacles with \DPV{}/\DMSP{} dynamics, we will use forward stochastic reachability to define a \occupfuntext{} and the \avoidsettext{} for a given time of interest, as done in~\cite{HomChaudhuriACC2017}.
We will seek grid-free, recursion-free, and computationally efficient algorithms to compute the \avoidsettext{} by exploiting known results to approximate convex and compact sets.
\begin{prob}
    Provide algorithms to approximate the \avoidsettext{}  $(\alpha\in[0,1])$ for a rigid body obstacle with \DPV{}/\DMSP{} dynamics:\label{prob_st:occupDPV_algo}
    \begin{enumerate}
        \item projection-based tight polytopic approximation, and 
        \item Minkowski sum-based overapproximation.
    \end{enumerate}
\end{prob}
\begin{customthm}{3.a}
    Provide sufficient conditions under which the \avoidsettext{} of a rigid body obstacle with \DPV{} dynamics is convex and compact.\label{prob_st:occupDPV_cvx_cmpt}
\end{customthm}
\begin{customthm}{3.b}
    Show that the \avoidsettext{} for a rigid body obstacle with \DMSP{} dynamics is covered by a union of convex and compact sets, the \avoidsettextNoAlpha{} of a collection of appropriately defined DPV dynamics.\label{prob_st:occupDMSP_algo}
\end{customthm}
For a rigid body obstacle with \DMSP{} dynamics, we will provide sufficient conditions under which the \occupfuntext{} is u.s.c, and the \avoidsettext{} is closed and bounded.
These sufficient conditions hold for \DPV{} dynamics as well.
In addition, we will also provide sufficient conditions under which the \occupfuntext{} is log-concave and the \avoidsettext{} is convex when the dynamics are \DPV{}.


\subsection{Real analysis}
\label{sub:real}

A function $f_u: \mathbb{R}^n \rightarrow \mathbb{R}$ is upper semi-continuous (u.s.c.) if its superlevel sets $ \{ x\in \mathbb{R}^n : f_u(x) \geq c\}$ are closed for every $c\in \mathbb{R}$~\cite[Defn. 7.13]{bertsekas_stochastic_1978}. 
Alternatively, $f_u( \overline{y})$ is u.s.c. if for every sequence $ \overline{y}_i \rightarrow \overline{y}$, we have $\limsup_{i \rightarrow \infty} f_u( \overline{y}_i) \leq f_u ( \overline{y})$~\cite[Lem. 7.13b]{bertsekas_stochastic_1978}.
For every u.s.c function $f_u(\cdot)$, $f_l( \overline{y})\triangleq -f_u( \overline{y})$ is lower semi-continuous (l.s.c.), i.e., $\liminf_{i \rightarrow \infty} (f_l( \overline{y}_i)) \geq f_l ( \overline{y})$ for every sequence $ \overline{y}_i \rightarrow \overline{y}$~\cite[Lem. 7.13a]{bertsekas_stochastic_1978}.

A function $f: \mathbb{R}^n \rightarrow \mathbb{R}$ is log-concave if $\log(f)$ is concave, and is quasiconcave if the sets $\{ \overline{y}\in \mathbb{R}^n : f( \overline{y}) \geq \beta\}$ are convex for all $\beta\in \mathbb{R}$~\cite[Sec. 3.4 and 3.5]{boyd_convex_2004}. 
By Heine-Borel theorem~\cite[Thm. 12.5.7]{TaoAnalysisII}, a set in $ \mathbb{R}^n$ is compact if and only if it is closed and bounded.


\subsection{Probability theory and Fourier transforms}
\label{sub:FT}

A random vector $\by$ is a measurable transformation defined in the probability space $(\mathcal{Y},\sigmaAlg(\mathcal{Y}), \Probby)$ with sample space $\mathcal{Y}\subseteq \mathbb{R}^p$, sigma-algebra $\sigmaAlg(\mathcal{Y})$, and probability measure $\Prob$ over $ \sigmaAlg( \mathcal{Y})$.
When $\by$ is absolutely continuous, $\by$ has a probability density function (PDF) $\psi_{\by}$ such that given $ \mathcal{G}\in \sigmaAlg( \mathcal{Y})$, $\Prob_{\by}\{\by\in \mathcal{G}\}=\int_{\mathcal{G}}\psi_{\by}(\overline{y})d\overline{y}$ where $\overline{y}\in \mathbb{R}^p$ and $\psi_{\by}$ is a non-negative Borel measurable function with $\int_{ \mathcal{Y}}\psi_{\by}(\overline{y})d\overline{y} = 1$~\cite[Ch. 1]{ChowProbability1997}.
The support of $\by$, denoted by $ \mathrm{supp}(\by)$, is the smallest closed subset of $\mathcal{Y}$ with probability of occurrence one.
Equivalently, from~\cite[Defn. A.5]{dharmadhikari1988unimodality},
\begin{align}
    \mathrm{supp}(\by)&=\left\{ \overline{y}\in \mathcal{Y}: \forall r>0,\ \Probby\{\by\in \Ball( \overline{y},r)\}>0\right\}. \label{eq:supp_defn}
\end{align}

From~\cite[Sec. 2.1]{dharmadhikari1988unimodality}, $\Probby$ is centrally symmetric if $\Probby\{\by\in \mathcal{G}\} = \Prob\{\by\in-\mathcal{G}\}$ for every $ \mathcal{G}\in \sigmaAlg( \mathcal{Y})$.
From~\cite[Sec. 2]{prekopa_logarithmic_1980}, a probability measure $\Probby$ is log-concave if for all convex $\mathcal{G}_A,\mathcal{G}_B\in\sigmaAlg (\mathcal{Y})$ and $\zeta\in[0,1]$, $\Probby\{ \by \in (\zeta \mathcal{G}_A\oplus (1-\zeta) \mathcal{G}_B)\} \geq {\Probby\{\by\in \mathcal{G}_A\}}^\zeta{\Probby\{\by\in \mathcal{G}_B\}}^{1-\zeta}$.
Given a convex borel set $ \mathcal{G}\in \sigmaAlg( \mathcal{Y})$ and a log-concave probability measure $\Probby$, the following function $h: \mathcal{Y} \rightarrow \mathbb{R}$ is log-concave,
\begin{align}
    h( \overline{c}) = \Probby\{ \by \in \{\overline{c}\} \oplus \mathcal{G}\}\label{eq:log-concaveProbSetMove}.
\end{align}

For some matrix $H\in \mathbb{R}^{n\times p}$, the stochasticity of the random vector $\bx = H \by$ may be characterized using Fourier transforms.
The \emph{characteristic function} (CF) of $\by\in \mathbb{R}^{p}$ is 
\begin{align} 
    \Psi_{\by}(\overline{\gamma})= \Exp_{\by}\left[\mathrm{exp}\left({j\overline{\gamma}^\top\by}\right)\right] &=\int_{\mathbb{R}^p}e^{j\overline{\gamma}^\top\overline{z}} \psi_{\by}(\overline{z})d\overline{z}= \mathscr{F}\left\{\psi_{\by}(\cdot)\right\}(-\overline{\gamma})\label{eq:cfun_def} 
\end{align} 
where $ \mathscr{F}\{\cdot\}$ denotes the Fourier transformation operator and $\overline{\gamma}\in \mathbb{R}^{p}$.
Since PDFs are absolutely integrable, every PDF has a unique CF~\cite[Pg. 382]{billingsley_probability_1995}.
The CF of the random vector $\bx$ is then given by
\begin{align}
    \Psi_{\bx}( \overline{\eta})&=\Psi_{\by}( H^\top \overline{\eta})\label{eq:Psibxfromby}
\end{align}
with Fourier variable $ \overline{\eta}\in \mathbb{R}^n$~\cite[Sec. 2.1]{VinodHSCC2017}. 
We can obtain the PDF of $\bx$ via inverse Fourier transform, provided $\Psi_{\bx}$ is absolutely integrable~\cite[Cor. 1.21]{SteinFourier1971}, square integrable~\cite[Thm. 2.4]{SteinFourier1971}, or Schwartz\footnote{Infinitely differentiable function on $ \mathbb{R}^p$ such that the function and its derivatives decrease rapidly~\cite[Sec. 6.2]{stein_fourier_2003}.}~\cite[Ch. 6, Thm. 2.4]{stein_fourier_2003},
\begin{align} 
    \psi_{\bx}(\overline{z})= \mathscr{F}^{-1}\left\{\Psi_{\bx}(\cdot)\right\}(-\overline{z}) &={\left(\frac{1}{2\pi}\right)}^n\int_{ \mathbb{R}^n}e^{-j\overline{\eta}^\top\overline{z}} \Psi_{\bx}(\overline{\eta})d\overline{\eta}.\label{eq:cfun_ift} 
\end{align} 
Here, $ \mathscr{F}^{-1}\{\cdot\}$ denotes the inverse Fourier transformation operator and $d\overline{\eta}$ is short for $d\eta_1d\eta_2\ldots d\eta_p$.
Alternatively, Levy's inversion theorem~\cite[Sec. 8.5, Thm. 1]{ChowProbability1997} may be used to compute $\int_{a_1}^{b_1}\ldots\int_{a_n}^{b_n} \psi_{\bx}( \overline{z})d \overline{z}$ for any $a_i<b_i \in \mathbb{R}$ with $i\in \mathbb{N}_{[1,n]}$ from its CF $\Psi_{\bx}$, the probability that $\bx$ lies in a hypercuboid.


\subsection{Properties of a rigid body}
\label{sub:rigid}

Let $ \mathcal{X} = \mathbb{R}^n$ denote the state space.
\begin{assum}
    The rigid body shape is a Borel set and has a non-zero Lebesgue measure.\label{assum:Borel}
\end{assum}
\begin{assum}
    The rigid bodies are only allowed to translate. \label{assum:obs}
\end{assum}
Assumption~\ref{assum:Borel} is typically satisfied by real-world problems since open and closed sets are Borel~\cite[Sec. 1.11]{RudinReal1987} and rigid body obstacle shapes are sets with positive ``volume''.
While Assumption~\ref{assum:obs} is common practice in motion planning problems~\cite[Sec. 4.3.2]{lavalle2006planning}~\cite{aoude2013probabilistically}, it excludes analysis of rigid body obstacles whose shape has a state-dependent orientation, for example, a unicycle obstacle with a non-centrally symmetric bounded shape that depends on the obstacle's heading.
However, by defining an overapproximative shape (that is Borel) that encompasses all attainable shapes, rigid body obstacles that do have state-dependent orientation can be accommodated. 
For the unicycle obstacle, we can use a large ball that contains the obstacle irrespective of the heading.

We then define the set of states ``occupied'' by the rigid body obstacle given the state of some point in the obstacle (say, center of mass) is $ \overline{c}\in \mathcal{X}$ as
\begin{align}
    \mathcal{O}( \overline{c}) &=\{\overline{z}\in \mathcal{X}: \hobs(\overline{z}-\overline{c})\geq 0\}\subseteq \mathcal{X}\label{eq:obs_rigidbody_defn}
\end{align}
using the zero super-level set of a known Borel-measurable function $\hobs: \mathcal{X} \rightarrow \mathbb R$.
For example, we define $\hobs(\overline{z}) = \frac{1}{2}-\|\overline{z}\|_{\infty}$ for an obstacle whose shape is an axis-aligned hyperbox $\Mybox( \overline{c}, 1)$ and define $\hobs(\overline{z}) = 1-\|\overline{z}\|_{2}$ for an obstacle whose shape is a unit sphere $\Ball( \overline{c}, 1)$.
\begin{lem}\label{lem:rigidBody}
    For an obstacle shape $ \mathcal{O}(\overline{y})$ with $ \overline{y}\in \mathcal{X}$, 
    {\renewcommand{\theenumi}{\alph{enumi}}
        \begin{enumerate}
            \item\label{lem:rigidBodyCenter} (translation invariance) $\mathcal{O}(\overline{y})= \{\overline{y}\} \oplus \mathcal{O}( \overline{0})$,
            \item\label{lem:rigidBody1} $-\mathcal{O}(-\overline{y}))=\{ \overline{z}\in \mathcal{X}: \overline{y}\in \mathcal{O}( \overline{z})\}$, and
            \item\label{lem:rigidBody2} $ 1_{(-\mathcal{O}(-\overline{y}))}( \overline{z})=1_{\mathcal{O}(\overline{0})}(\overline{y}-\overline{z})$.
        \end{enumerate}
    }
\end{lem}
\begin{proof}
 See Appendix~\ref{app:proof_lem_rigid_body}.
\end{proof}
Lemma~\ref{lem:rigidBody} provides some useful properties of the set $ \mathcal{O}(\overline{c})$ \emph{independent} of the geometric properties of the rigid body like closedness, convexity, and boundedness.
Moreover, Lemma~\ref{lem:rigidBody}\ref{lem:rigidBodyCenter} shows that it is sufficient to impose geometric restrictions only on $\mathcal{O}(\overline{0})$, due to Assumption~\ref{assum:obs}.

The obstacle shape $ \mathcal{O}( \overline{0})$ is centrally symmetric set if $\mathcal{O}( \overline{0}) = - \mathcal{O}( \overline{0})$~\cite[Sec. 2.1]{dharmadhikari1988unimodality}.
Here, $(-\mathcal{O}( \overline{c}))$ is the reflection of set $ \mathcal{O}( \overline{c})$ about origin,
\begin{align}
    -\mathcal{O}( \overline{c})&=\{\overline{z}\in \mathcal{X}: -\overline{z}\in \mathcal{O}(\overline{c})\}= - I_n \mathcal{O}( \overline{c}).\label{eq:obs_rigidbody_defn_reflect}
\end{align}

\subsection{Computation of tight polytopic approximations for arbitrary convex and compact sets}
\label{sub:tight_polytopic_approx}

Tight polytopic over/underapproximation of convex and closed sets have been well studied~\cite[Ex. 2.25]{boyd_convex_2004}~\cite[Ch. 2]{webster1994convexity}. 
Consider an arbitrary convex and closed set 
\begin{align}
    \mathcal{L} = \{ \overline{y}\in \mathbb{R}^n: f( \overline{y}) \geq \beta\} \label{eq:E_def}
\end{align}
for a known u.s.c and quasiconcave function $f: \mathbb{R}^n \rightarrow \mathbb{R}$ and known $\beta\in \mathbb{R}$.
Let the maxima of $f$ be larger than $\beta$, \emph{i.e.}, $ \mathcal{L}$ is non-empty.

Given $K>0$ points external to $ \mathcal{L}$, $ \overline{p}_i\in \mathbb{R}^n\setminus \mathcal{L},\ i\in \mathbb{N}_{[1,K]}$.
We project $ \overline{p}_i$ onto $ \mathcal{L}$ by solving \eqref{prob:projection_problem} for each $i\in \mathbb{N}_{[1,K]}$~\cite[Sec. 8.1]{boyd_convex_2004},
\begin{align}
     \begin{array}{rl}
         \underset{ \overline{y}\in  \mathbb{R}^n}{\mathrm{minimize}}& {\Vert \overline{y} - \overline{p}_i \Vert}_2\\
         \mbox{subject to}& f( \overline{y})\geq \beta
    \end{array} \label{prob:projection_problem}.
\end{align}
For each $p_i$, \eqref{prob:projection_problem} has a unique optimal solution, and we denote this projection point as $\ProjL( \overline{p}_i) \in \mathcal{L}$. 
We also associate a hyperplane \eqref{eq:hyperplane_defn} with each $p_i$,
\begin{align}
    \overline{a}_i^\top ( \overline{y} - \ProjL( \overline{p}_i))&\leq 0\mbox{ with } \overline{a}_i = \overline{p}_i - \ProjL( \overline{p}_i). \label{eq:hyperplane_defn}
\end{align}
Algorithm~\ref{algo:tight} solves \eqref{prob:projection_problem} for every $ \overline{p}_i$ to compute two polytopes $\Linner(K)$ and $\Louter(K)$,
\begin{align}
    \Linner(K)&= \mathrm{ConvexHull}( \ProjL(\overline{p}_1),\ldots, \ProjL(\overline{p}_K)) \label{eq:Linner_defn},\\
    \Louter(K)&= \cap_{i=1}^K\{ \overline{y}\in  \mathbb{R}^n: \overline{a}_i^\top( \overline{y} - \ProjL( \overline{p}_i))\leq 0\}.\label{eq:Louter_defn}
\end{align}
\begin{lem}\label{lem:algo_proof}
    For a convex, closed, and non-empty set $ \mathcal{L}$ and $K$ points $ \overline{p}_i\not\in \mathcal{L},\ i\in \mathbb{N}_{[1,K]}$, Algorithm~\ref{algo:tight} provides tight polytopic over- and underapproximation, i.e., $\Linner(K)\subseteq \mathcal{L}\subseteq\Louter(K)$.
\end{lem}
\begin{proof}
    See Appendix~\ref{app:proof_algo_proof}.
\end{proof}

The computation of $ \overline{p}_i$ is easy when $ \mathcal{L}$ is bounded.
For some $ \overline{y}\in \mathcal{L}$, $\exists r>0$ such that $ \mathcal{L}\subseteq \Ball( \overline{y}, r)$.
We can now obtain the desired $K$ points that lie outside $ \mathcal{L}$ by sampling the surface of this ball, denoted by $\partial \Ball( \overline{y}, r)$.
For $n\in\{2,3\}$, we can uniformly discretize the boundary of a bounding circle/sphere respectively to obtain $ \overline{p}_i$.
For higher dimensions, we obtain $ \overline{p}_i$ by sampling an appropriately dimensioned standard normal distribution, and normalizing the samples to force it to lie on the surface of the bounding hypersphere~\cite{harman2010decompositional}. 
Computation of the bounding radius $r$ is equivalent to the computation of the diameter of a compact set~\cite[Sec. 2.2]{webster1994convexity}, or may be set to a sufficiently large value.

\begin{algorithm}
    \caption{Tight polyhedral approximations of $ \mathcal{L}$ \eqref{eq:E_def}}\label{algo:tight}
    \begin{algorithmic}[1]    
        \Require{u.s.c and quasiconcave $f$, $\beta\leq\max_{ \overline{y}\in \mathbb{R}^n} f( \overline{y})$,  $K>0$, points $ \overline{p}_i\not\in \mathcal{L}$ for $i\in \mathbb{N}_{[1,K]}$}
        \Ensure{$\Linner(K),\Louter(K)$ s.t. $\Linner(K)\subseteq \mathcal{L}\subseteq\Louter(K)$}     
        \State Solve \eqref{prob:projection_problem} for every $i\in \mathbb{N}_{[1,K]}$ to obtain $\ProjL( \overline{p}_i)$ 
        \State Compute $\Linner(K)$ using the convex hull as in \eqref{eq:Linner_defn} 
        \State Compute $\Louter(K)$ using the supporting hyperplanes as in \eqref{eq:Louter_defn}
  \end{algorithmic}
\end{algorithm}


\section{System definition}
\label{sec:sys}

\subsection{Discrete-time Markovian switched system with affine parameter-varying stochastic subsystems (DMSP)}
\label{sub:DMSP}

The \DMSP{} dynamics consists of affine parameter varying subsystems where the evolution of the parameter depends on the discrete mode of the system.
The subsystems are affine since the control input is assumed to follow a known\footnote{Allows incorporation of drift terms into the obstacle dynamics.} open-loop policy.
The discrete mode switches in a Markovian time-dependent switching.
Given a finite time horizon $N\in \mathbb{N}$, the \DMSP{} dynamics for $t\in \mathbb{N}_{[0,N-1]}$ are given by
\begin{subequations}
\begin{align}
    \bsig_{t+1}&= \mathscr{P}_{\bsig}(t,\bsig_t)\label{eq:sys_disc}\\
    \blam_{t+1}&=l_{\blam}(\bsig_{t+1},t, \blam_{t}) \label{eq:param_cts}\\
    \bx_{t+1}&=A(\blam_t)\bx_t+B(\blam_t)\overline{u}_t +
    F(\blam_t)\bw_t \label{eq:sys_cts}
\end{align}\label{eq:DMSP_sys}%
\end{subequations}
with discrete state $\bsig_t\in \mathcal{Q}$ (a finite subset of $\mathbb{R}$), parameter $\blam_t\in \mathcal{P}$ (a finite subset of $\mathbb{R}^l$), continuous state $\bx_t\in \mathcal{X} = \mathbb{R}^n$, $\bfs_t={[{\bsig_t}\ {\bx_t}^\top]}^\top\in \mathcal{S}= \mathcal{Q}\times \mathcal{X}$, input $\overline{u}_t\in \mathcal{U}\subseteq \mathbb{R}^m$, stochastic disturbance $\bw_t\in \mathcal{W}\subseteq \mathbb{R}^p$, appropriately defined parameter-varying discrete state dependent matrices $A,B,$ and $F$, time-dependent Markovian stochastic map $ \mathscr{P}_{\bsig}$, and (possibly nonlinear) known Borel-measurable function $l_{\blam}: \mathcal{Q}\times\mathbb{N}_{[0,N-1]}\times  \mathcal{P} \rightarrow \mathcal{P}$ that translates the discrete state to the parameter.
An example of a time-dependent Markovian switching law $ \mathscr{P}_{\bsig}$ is
\begin{align}
    \bsig_{t+1}&=\begin{cases}
        \begin{array}{ll}
            \mathscr{M}(\bsig_t)\ & t=k\tausw,\ k\in \mathbb{N}_{[1,N-1]} \\
            \bsig_t & \mbox{otherwise}
        \end{array}
    \end{cases}\label{eq:MarkovSwitchLaw}
\end{align}
with $ \mathscr{M}$ used to represent a time-invariant\footnote{The transition probabilities are stationary~\cite[Sec. 8]{billingsley_probability_1995}.} Markov chain defined on $ \mathcal{Q}$ with a transition probability matrix $M\in \mathbb{R}^{\vert \mathcal{Q} \vert \times \vert \mathcal{Q} \vert}$. 
Here, the discrete state switches every $\tausw\in \mathbb{N}\setminus\{0\}$ time steps with the switching governed by a Markov chain $ \mathscr{M}$.

\begin{figure}
\def\firstcircle{(0,0) circle (0.5cm)}
\def\secondcircle{(0:0cm) circle (1.2cm)}
\def\thirdcircle{(0:0cm) circle (1.9cm)}

\def\firstcircle{(0.4\linewidth,0) rectangle (6,2)}
\def\secondcircle{(0.4\linewidth,0) circle (2.5cm)}

\centering
\begin{tikzpicture}
    \node[align=center, text width = 10cm] (DPVtext) at (0.4\linewidth,0) {Discrete-time  stochastic parameter-varying affine systems (\DPV{}) $\equiv$ Discrete-time stochastic time-varying affine systems (DTV)};
    \node[draw, fit=(DPVtext), shape=rectangle, rounded corners] (DPVell) {};
    \node[align=center, text width = 7cm, above of=DPVtext, yshift=0.15cm] (MJAStext) {Markov jump affine systems};
    \node[align=center, text width = 7cm, above of=MJAStext] (DMSPtext) {Discrete-time Markovian switched systems with \DPV{} subsystems (\DMSP{})};
    \node[align=center, above of=DMSPtext] (DTSHStext) {Discrete-time stochastic hybrid systems (DTSHS)};
    \node[draw, rounded corners, fit=(MJAStext)(DPVell),shape=rectangle] (MJASrect) {};
    \node[draw, rounded corners, fit=(MJASrect)(DMSPtext),shape=rectangle] (DMSPrect) {};
    \node[draw, rounded corners, fit=(DMSPrect)(DTSHStext),shape=rectangle] (DTSHSrect) {};
\end{tikzpicture}
\caption{Relationship between various system models\label{fig:sys_type}}
\end{figure}
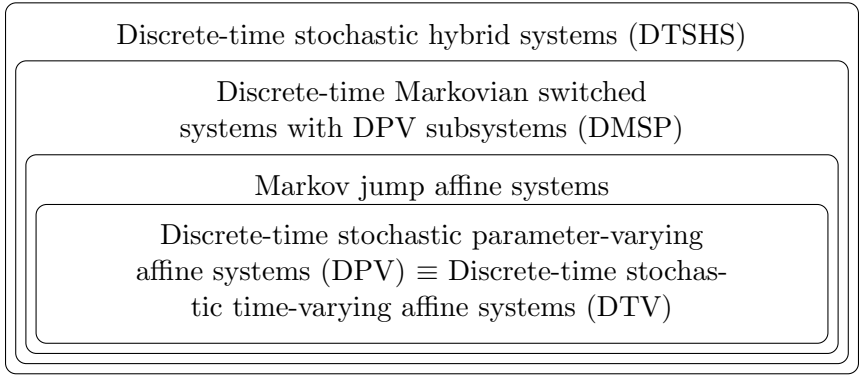

The \DMSP{} is a DTSHS with a discrete transition kernel that is independent of the continuous state (permitting only time-dependent switching) and has an identity reset map~\cite{Abate2008}.
The evolution of the \DMSP{} therefore may also be defined using the execution of the DTSHS~\cite[Def. 3]{Abate2008}. 
The \DMSP{} dynamics described in \eqref{eq:DMSP_sys} includes three special classes of systems: 
\begin{enumerate}
    \item \emph{discrete-time stochastic affine time-varying system} (DTV) with non-stochastic parameter $\overline{\lambda}_t=t$ and $\mathcal{Q}=\emptyset$,
    \item \emph{discrete-time stochastic affine parameter-varying system} (DPV), with non-stochastic parameter $\overline{\lambda}_t=t$ that evolves according to $l_{\lambda}(\cdot,t,\overline{\lambda}_t)=g_{\lambda}(t,\overline{\lambda}_{t})$ for some known function $g_{\lambda}: \mathbb{N}_{[0,N]}\times \mathcal{P} \rightarrow \mathcal{P}$, $g_{\lambda}(0,\cdot)= \overline{\lambda}_0$, and $\mathcal{Q}=\emptyset$~\cite{hoffmann_survey_2015, LimPhD}, and
    \item \emph{Markov jump affine system}, with $\blam_t=\bsig_t$ and $ \mathscr{P}_{\bsig}(t,\bsig_t)$ is of the form \eqref{eq:MarkovSwitchLaw} with $\tausw=1$~\cite{costa_discrete-time_2005}. 
\end{enumerate}
Note that the parameter $\blam_t$ defined in \eqref{eq:DMSP_sys} is non-stochastic in \DPV{} systems and DTV systems, and hence the parameter is denoted by $ \overline{\lambda}_t$.
Also, the class of \DPV{} systems and DTV systems defined here are equivalent.
Clearly, the DTV systems are \DPV{} systems with time as the parameter.
To express \DPV{} systems as DTV systems, we use $G(t): \mathbb{N}_{[0,N]} \rightarrow \mathcal{P}$ which denotes the known parameter trajectory corresponding to the parameter dynamics $\overline{\lambda}_{t+1}=g_\lambda(t,\overline{\lambda}_{t})$ initialized to  $ \overline{\lambda}_0$. 
Using $G$, one can convert a \DPV{} system into a DTV system at the expense of more complicated definitions for $A(\cdot)$, $B(\cdot)$, and $F(\cdot)$.
One can also write \eqref{eq:DMSP_sys} as a Markov affine jump system-like model by combining \eqref{eq:sys_disc} and \eqref{eq:param_cts} to define a single discrete state that evolves under a more sophisticated stochastic map $ \mathscr{P}_{\bsig}$. 
However, unlike Markov affine jump system models, \DMSP{} dynamics permit $\tausw\neq 1$.
The relationship between all these system formulations is described in Figure~\ref{fig:sys_type}.

\subsubsection{Stochasticity of the discrete state}

From \eqref{eq:sys_disc}, $\bsig_\tau$ is random variable for every $\tau\in \mathbb{N}_{[1,N]}$.
The Markovian stochastic map $ \mathscr{P}_{\bsig}$ characterizes the (possibly time-varying) discrete transition kernel $\ProbSigOnly\{\bsig_{t+1}\vert\bsig_t\}$ for $t\in \mathbb{N}_{[0,N-1]}$.
We denote the probability measure of the concatenated discrete state random vector $\overline{\bsig}_\tau={[\bsig_1\ \bsig_2\ \ldots\ \bsig_\tau]}^\top\in \mathcal{Q}^\tau$ as $\ProbOvbsigTau$. 
Given an initial discrete state $q_0\in \mathcal{Q}$ and $\ProbSigOnly\{\bsig_{t+1}\vert\bsig_t\}$, we define the probability of $ \overline{\bsig}_\tau$ being $ \overline{q}_\tau={[q_1\  q_2\ \ldots\ q_{\tau}]}^\top\in \mathcal{Q}^\tau$ as
\begin{align}
    \ProbOvbsigTau\{\overline{\bsig}_\tau=\overline{q}_\tau\}&=\prod_{t=0}^{\tau-1}\ProbSigOnly\{\bsig_{t+1}=q_{t+1}\vert\bsig_t=q_{t}\}.\label{eq:ProbBarSig}
\end{align} 
In contrast to the discrete transition kernel $T_q$ in~\cite{Abate2008}, $\ProbSigOnly$ is independent of the continuous state.
As is typically done for Markov chains, a probability measure $\ProbbsigTau$ and a probability mass function $\psi_{\bsig}[q;\tau,q_0]$ for the discrete state $\bsig_\tau$ at time $\tau$ may be defined using \eqref{eq:ProbBarSig}.

\subsubsection{Stochasticity of the parameter}

The definition \eqref{eq:param_cts} ensures that the parameter $\blam_t$ is a random vector, and the current discrete state influences the evolution of the current continuous state through the parameter.
Since $l_{\blam}$ is a Borel-measurable function, we can define\footnote{We omit the explicit mention of $q_0$ and $ \overline{\lambda}_0$ in $\Llambdatau$ for brevity.} a Borel-measurable map $\Llambdatau: \mathcal{Q}^t \rightarrow \mathcal{P}^t$ given $ \overline{\lambda}_0$, $q_0$ and $\overline{\bsig}_\tau$ for any $t\in \mathbb{N}$.
Given $ \overline{\bsig}_\tau$, we define $ \overline{\blam}_\tau = \Llambdatau( \overline{\bsig}_\tau) = {[ \blam_1^\top\ \blam_2^\top\ \ldots\ \blam_\tau^\top]}^\top\in \mathcal{P}^\tau$ as the concatenated random vector describing the parameter values taken in $ \mathbb{N}_{[1,\tau]}$.
The stochasticity of the random vector $\overline{\blam}_\tau$ is induced from $\ProbOvbsigTau$ through $\Llambdatau$ and $\overline{\lambda}_0$, and its realizations is denoted by $ \overline{\Lambda}_\tau = \Llambdatau( \overline{q}_\tau)$.

\subsubsection{Stochasticity of the continuous state}
The disturbance in the subsystem dynamics \eqref{eq:sys_cts} is an independent and identically distributed (IID) random process with known parameter-independent probability density $\psi_{\bw}$ and well-defined probability space $( \mathcal{W}, \sigmaAlg( \mathcal{W}), \Prob_{\bw})$.
Without loss of generality, we assume that $ \mathrm{supp}(\bw) = \mathcal{W} \subseteq \mathbb{R}^p$.
We denote the initial continuous state as $ \overline{x}_0\in \mathcal{X}$ and the known concatenated vector inputs as $\UN={[\overline{u}_{N-1}^\top\ \overline{u}_{N-2}^\top\ \ldots\ \overline{u}_1^\top\ \overline{u}_0^\top]}^\top\in \mathcal{U}^N$.
Note that due to the hierarchy present in \eqref{eq:DMSP_sys}, the continuous state dynamics at the first time step is influenced by the initial parameter value $ \overline{\lambda}_0$ and not the initial discrete state $ \overline{q}_0$.
From \eqref{eq:sys_cts}, $\bx_\tau$ is a random vector whose stochasticity depends on $\overline{x}_0$, $\overline{\lambda}_0$, $\overline{\Lambda}_{\tau-1} = \LlambdatauMinusOne( \overline{q}_{\tau-1})$, and $ \Utau$. 
We denote the probability measure associated with $\bx_t$ is denoted by $\Probbx$. 
Note that the continuous state is preserved under a discrete state switching, i.e., an identity reset map.

\begin{rem}\label{rem:causality}
    We parameterize $\Probbx$ by $\Utau$ even when the control policy is known until $N-1>\tau-1$ is known to emphasize causality.
\end{rem}

\subsubsection{Stochasticity of the state of \eqref{eq:DMSP_sys}}

As in~\cite[Sec. 3]{davis_piecewise-deterministic_1984}, the state space of \eqref{eq:DMSP_sys} is $ \mathcal{S} = \mathcal{Q}\times \mathcal{X}$, with the
 corresponding sigma-algebra $ \sigmaAlg ( \mathcal{S})$
\begin{align}
    \sigmaAlg(\mathcal{S})&= \sigmaAlg\left(\cup_{q\in \mathcal{Q}} \left\{(q, \mathcal{G}_X(q)): q\in \mathcal{Q}, \mathcal{G}_X(q) \in \sigmaAlg( \mathcal{X})\right\}\right).\label{eq:sigmaAlgDMSP}
\end{align}
Additionally, for any $q_1,q_2\in \mathcal{Q},\ q_1\neq q_2$ and  non-empty $\mathcal{G}_{X,1}(q_1),\ \mathcal{G}_{X,2}(q_2)\in \sigmaAlg( \mathcal{X})$, the events $(\{q_1\}\times \mathcal{G}_{X,1}(q_1))$ and $(\{q_2\}\times \mathcal{G}_{X,2}(q_2))$ are disjoint. 

We denote the set of discrete states present in $ \mathcal{G}_S$ as the projection of a set $ \mathcal{G}_S \in\sigmaAlg( \mathcal{S})$ onto $ \mathcal{Q}$,
\begin{align}
    \Projq( \mathcal{G}_S)&= \{q\in \mathcal{Q}: (q, \emptyset) \subseteq \mathcal{G}_S\}.\label{eq:ProjectionDefn}
\end{align}
\begin{lem}{\textbf{(Events in $ \sigmaAlg( \mathcal{S})$)}}
   Every Borel set  $ \mathcal{G}_S\in\sigmaAlg( \mathcal{S})$ can be written as
    \begin{align}
        \mathcal{G}_S&=\cup_{q\in \Projq( \mathcal{G}_S)} (\{q\}\times \mathcal{G}_X(q))\label{eq:GsDef}
    \end{align}
    for some $\mathcal{G}_X(q) \in \sigmaAlg(\mathcal{X})$ for every $q\in \Projq( \mathcal{G}_S)$. 
    Moreover, the representation \eqref{eq:GsDef} partitions $ \mathcal{G}_S$. \label{lem:BorelS}
\end{lem}
Lemma~\ref{lem:BorelS} states that the events of interest for the system \eqref{eq:DMSP_sys} can be decomposed into a collection of events, each event consiting of discrete state taking a particular value and an associated Borel set in which the continuous state must lie, and this decomposition is disjoint.
A similar approach was taken in~\cite[Sec. 2]{Abate2008} to define the sigma-algebra of the hybrid state in DTSHS.
We will denote the known initial state of \DMSP{} by $\overline{s}_0={[q_0\ {\overline{x}_0}^\top]}^\top\in \mathcal{S}$.

For an initial state $ \overline{x}_0$, initial parameter value $ \overline{\lambda}_0$, and known parameter sequence $ \overline{\Lambda}_{N-1}$, the \DPV{} dynamics are given by 
\begin{align}
    \bx_{t+1}&=A(\overline{\lambda}_t)\bx_t+B(\overline{\lambda}_t)\overline{u}_t + F(\overline{\lambda}_t)\bw_t.\label{eq:DPV_sys}
\end{align}%



\subsection{Modified unicycle dynamics: \DMSP{} and \DPV{} models}
\label{sub:unicycle}

We consider the unicycle model discretized in time with sampling time $T_s$. 
Additionally, we enforce the following
\begin{enumerate}
    \item heading velocity $\bv_t$ is a random variable with a known probability density $\psi_v$, and
    \item turning rate $\bom_t$ is a discrete-valued random variable whose evolution is modelled as a time-dependent Markov chain.
\end{enumerate}
The modified unicycle dynamics are given by
\begin{subequations}
\begin{align}
    \left[\begin{array}{c}
            {(\bp_{t+1})}_{x} \\
            {(\bp_{t+1})}_{y} \\
            \bth_{t+1}
    \end{array}\right]&=\left[\begin{array}{c}
        {(\bp_t)}_{x} \\
        {(\bp_t)}_{y} \\
        \bth_t
    \end{array}\right] + T_s\left[\begin{array}{c}
            \bv_t \cos\big(\bth_t\big) \\
            \bv_t \sin\big(\bth_t\big) \\
            \bom_t
    \end{array}\right] \label{eq:obs_dyn_cts}\\
    \bom_{t+1}&=\begin{cases}
        \begin{array}{ll}
            \mathscr{M}_{\bom}(\bom_t)\ & t=k\tausw,\ k\in \mathbb{N}_{[1,N-1]} \\
            \bom_t & \mbox{otherwise}
    \end{array}
    \end{cases} \label{eq:obs_omega}\\
    \bv_t&\sim \psi_{\bv} \label{eq:obs_speed}
\end{align}\label{eq:obs_nonlin_sys}%
\end{subequations}
with location $\bp_t={[{(\bp_t)}_x\ {(\bp_t)}_y]}^\top\in \mathbb{R}^2$, heading $\bth_t\in(-\pi,\pi]$, and turning rate $\bom_t\in \mathcal{Q}$ updated every $\tausw$ time steps based on the Markov chain $ \mathscr{M}_{\bom}$.
The system \eqref{eq:obs_nonlin_sys} has no inputs, but can be extended to include terms that describe the known drifts in the model.  

Defining the continuous state, the parameter, and the discrete as the current location $\bx_t = \bp_t$, the current heading $\blam_t=\bth_t$, and the previous turning rate $\bsig_t = \bom_{t-1}$ of the unicycle respectively, we recast the nonlinear system \eqref{eq:obs_nonlin_sys} into the \DMSP{} dynamics with 
\begin{subequations}
\begin{align}
    \bsig_{t+1}&=\begin{cases}
        \begin{array}{ll}
            \mathscr{M}_{\bom}(\bsig_t)\ & t=k\tausw+1,\ k\in \mathbb{N}_{[1,N-1]} \\
            \bsig_t & \mbox{otherwise}
    \end{array}
    \end{cases} \label{eq:obs_omega_sig}\\
    \blam_{t+1}& = \blam_{t} + T_s \bsig_{t+1} \label{eq:param_cts_obs}\\ 
    \bx_{t+1}&=\bx_t+
    B_O(\blam_t)\bv_t \label{eq:obs_unicycle_dyn}\\
    \bv_t&\sim \psi_{\bv} \label{eq:obs_speed_DMSP}
\end{align}\label{eq:obs_DMSP}%
\end{subequations}
where $B_O(\blam_t) = T_s{\left[\cos\big(\blam_t\big)\ \sin\big(\blam_t\big)\right]}^\top$.
We define the initial continuous state of \eqref{eq:obs_DMSP} as the initial unicycle location $ \overline{x}_0\in \mathcal{X}$, the initial parameter value of \eqref{eq:obs_DMSP} as the initial unicycle heading $\lambda_0\in(-\pi,\pi]$, the initial discrete state of \eqref{eq:obs_DMSP} as the initial unicycle turning rate $\omega_0\in \mathcal{Q}$, and the initial state of \eqref{eq:obs_DMSP} as $ \overline{s}_0 = {[q_0\ \bar{z}_0]}^\top\in \mathcal{Q}\times\mathcal{X}$.

Consider a known sequence of turning rates $ \overline{q}_\tau$ where $q_t=\omega_{t-1}$ for $t\in \mathbb{N}_{[1,\tau]}$.
The unicycle location is given by the following \DPV{} dynamics,
\begin{subequations}
    \begin{align}
    \lambda_{t+1}&=\lambda_t+ q_{t+1} T_s \label{eq:obs_unicycle_dyn_param_DPV}\\
    \bx_{t+1}&=\bx_t + B_O(\lambda_t)\bv_t \label{eq:obs_unicycle_dyn_DPV}\\
    \bv_t&\sim \psi_{\bv}.\label{eq:obs_speed_DPV}
\end{align}\label{eq:obs_DPV}%
\end{subequations}%
\section{Forward stochastic reachability}
\label{sec:FSR}

Forward stochastic reachability of a system characterizes the stochasticity of the state of a given stochastic system at a future time of interest.
It provides the probability measure associated with the state, known as the \emph{forward stochastic reach probability measure} (FSRPM), and the support of the state, known as the \emph{forward stochastic reach set} (FSR set), at the time of interest.

\subsection{Forward stochastic reachability for \DPV{} dynamics}
\label{sub:FSR_DPV}

Consider the \DPV{} \eqref{eq:DPV_sys} initialized to $ \overline{x}_0, \overline{\lambda}_0$, $\overline{\Lambda}_{N-1}$, and $\UN$.
Similarly to $ \UN$, define $\overline{\bw}_{N-1}={[\bw_{N-1}^\top\ \bw_{N-2}^\top\ \ldots\ \bw_{1}^\top\ \bw_{0}^\top]}^\top\in \mathcal{W}^N$ as the concatenated random vector of $\bw_t$ with probability measure $\ProbOvBwNNoTail$, easily defined using $\psi_{\bw}$ by the IID assumption on $\bw$.
Here, the FSRPM is the probability measure $\ProbbxLambda$ and the FSR set $\FSRsetDPV$ is the support of the state random vector $\bx_{\tau}$ at time $\tau$ respectively.
If a non-negative Borel function $\psibxzLambda$ exists, such that $\int_{ \mathcal{X}}\psibxzLambda d\overline{z} = 1$, and for any $ \mathcal{G}_X\in \sigmaAlg(\mathcal{X})$,
\begin{align}
    \ProbbxLambda\{\bx_\tau\in \mathcal{G}_X\} =\int_{ \mathcal{G}_X} \psibxzLambda dz.\label{eq:FSRPD_defn}
\end{align}
then $\psibxzLambda$ is the \emph{forward stochastic reach probability density} (FSRPD) of the \DPV{}.

The definition of the matrices $ \mathscr{A}$, $\mathscr{C}_U$, and $ \mathscr{C}_W$ in \eqref{eq:matrices_skaf} is inspired from~\cite[eq. (3)]{SkafTAC2010}.
\begin{subequations}
\begin{align}
    \mathscr{A}(i,j;\overline{\Lambda}_{\tau-1}, \overline{\lambda}_0)&=\left(\prod_{t=i}^{j-1} A( \overline{\lambda}_{t})\right)\in \mathbb{R}^{n\times n}, \mbox{ with }\mathscr{A}(i,i;\overline{\Lambda}_{\tau-1}, \overline{\lambda}_0)=I_n,\mbox{ and } i,j\in \mathbb{N},i<j \label{eq:ConcatA} \\ 
    \mathscr{C}_U (\tau;\overline{\Lambda}_{\tau-1}, \overline{\lambda}_0) &=[\mathscr{A}({\tau},{\tau};\overline{\Lambda}_{\tau}, \overline{\lambda}_0)B( \overline{\lambda}_{\tau-1})\ \mathscr{A}({\tau-1},{\tau};\overline{\Lambda}_{\tau}, \overline{\lambda}_0)B(\overline{\lambda}_{\tau-2})\ \ldots\ \mathscr{A}(1,{\tau};\overline{\Lambda}_{\tau}, \overline{\lambda}_0)B( \overline{\lambda}_{0})] \in \mathbb{R}^{n\times (m\tau)},\label{eq:Cinput}\\
    \mathscr{C}_W (\tau;\overline{\Lambda}_{\tau-1}, \overline{\lambda}_0)&=[\mathscr{A}({\tau},{\tau};\overline{\Lambda}_{\tau}, \overline{\lambda}_0)F( \overline{\lambda}_{\tau-1})\ \mathscr{A}({\tau-1},{\tau};\overline{\Lambda}_{\tau}, \overline{\lambda}_0)F(\overline{\lambda}_{\tau-2})\ \ldots\ \mathscr{A}(1,{\tau};\overline{\Lambda}_{\tau}, \overline{\lambda}_0)F( \overline{\lambda}_{0})] \in \mathbb{R}^{n\times (p\tau)}.\label{eq:Cdist}
\end{align}\label{eq:matrices_skaf}
\end{subequations}
We compactly write the state $\bx_{\tau}$ at a time of interest $\tau\geq 1$ as an affine transformation of $ \overline{\bw}_{\tau-1}$ by separating the elements that evolve under the influence of the stochastic disturbance from those that evolve deterministically, i.e.,
\begin{align}
    \bx_\tau &=\xunpert+ \mathscr{C}_W(\tau;\overline{\Lambda}_{\tau-1}, \overline{\lambda}_0)\bwtau\label{eq:traj_LTV}
\end{align}
with $\xunpert$ as the unperturbed state (disturbance free), 
\begin{align}
    \xunpert&= \mathscr{A}(0,\tau;\overline{\Lambda}_{\tau-1},\overline{\lambda}_0)\overline{x}_0 +\mathscr{C}_U(\tau;\overline{\Lambda}_{\tau-1},\overline{\lambda}_0)\Utau \label{eq:xunpert}.
\end{align}
We will omit references to $ \overline{x}_0, \overline{\lambda}_0$, $\overline{\Lambda}_{N-1}$, and $ \Utau$ for brevity.
We define the random vector $\bW_\tau =\mathscr{C}_W(\tau;\cdot) \overline{\bw}_{\tau -1}$ with the probability measure $\ProbbWtau$ such that $\forall \mathcal{G}\in \sigmaAlg( \mathcal{W}^\tau)$,
\begin{align}
    \ProbbWtau\{\bW_\tau \in \mathcal{G}\}=\ProbOvBwtauNoTail\{\mathscr{C}_W(\tau;\cdot) \overline{\bw}_{\tau -1} \in \mathcal{G}\} &=\ProbOvBwtauNoTail\{ \overline{\bw}_{\tau-1}\in\{ \overline{y}\in \mathcal{W}^\tau:\mathscr{C}_W(\tau;\cdot) \overline{y} \in \mathcal{G}\}\}.\label{eq:ProbbW_defn}
\end{align}
The probability measure $\ProbbWtau$ also has an associated PDF $\psibWtauFull$.
Theorem~\ref{thm:FSRPM_def_DPV} follows from \eqref{eq:traj_LTV} and \eqref{eq:ProbbW_defn}, and Corollary~\ref{corr:FSRPD_def_DPV} follows from \eqref{eq:FSRPD_defn}.
\begin{thm}{\textbf{(FSRPM for \DPV{})}}\label{thm:FSRPM_def_DPV}
    For any time instant $\tau\in\mathbb{N}_{[1,N]}$ and some Borel set $ \mathcal{G}_X \in \sigmaAlg( \mathcal{X})$, the FSRPM $\ProbbxNoTail$ of the \DPV{} is given by 
    \begin{align}
        \ProbbxNoTail \{\bx_\tau \in \mathcal{G}_X\} &=\ProbOvBwtauNoTail\left\{ \bW_\tau\in \left( \mathcal{G}_X \oplus\left\{-\xunpertNoTail\right\}\right)\right\}.\nonumber
\end{align}
\end{thm}
\begin{corr}{\textbf{(FSRPD for \DPV{})}}\label{corr:FSRPD_def_DPV}
    Given the probability density $\psibWtauNoTail$ for $\bW_\tau$ for any time instant $\tau\in\mathbb{N}_{[1,N]}$ and $\overline{z}\in \mathcal{X}$, the FSRPD $\psibxzNotail$ of the \DPV{} is 
    \begin{align}
        \psibxzNotail&=\psibWtau\left( \overline{z} - \xunpertNoTail;\tau,\cdot\right). \nonumber
\end{align}
\end{corr}%
\begin{rem}
    We can relax the IID assumption on $\bw$ if $\ProbOvBwNNoTail$ is known.
\end{rem}
For a \DPV{} with a Gaussian disturbance $\bw\sim \mathcal{N}(\overline{\mu}_{\bw}, \Sigma_{\bw})$, the random vector $ \overline{\bw}_{\tau -1}$ is Gaussian with mean $[\overline{\mu}_{\bw}\ \ldots\ \overline{\mu}_{\bw}]=1_{\tau\times 1}\otimes \overline{\mu}_{\bw}\in \mathbb{R}^{\tau p}$ and covariance matrix $I_\tau \otimes \Sigma_{\bw} \in \mathbb{R}^{(\tau p)\times (\tau p)}$ due to the IID assumption.
Hence, $\bW_\tau=\mathscr{C}_W(\tau;\cdot) \overline{\bw}_{\tau -1}$ is a Gaussian random vector~\cite[Sec. 9.2]{GubnerProbability2006} with PDF,
\begin{subequations}
\begin{align}
    \psibWtau(\overline{y};\tau,\cdot)&\sim \mathcal{N}( \overline{\mu}_{\bW_\tau}, \Sigma_{\bW_\tau}) \label{eq:Gauss_psibWApp_pdf},\mbox{ with } \\
    \overline{\mu}_{\bW_\tau}&= \mathscr{C}_W(\tau;\cdot)(1_{\tau\times 1}\otimes \overline{\mu}_{\bw}),\label{eq:Gauss_psibWApp_mu} \\
    \Sigma_{\bW_\tau}&= \mathscr{C}_W(\tau;\cdot)(I_\tau \otimes \Sigma_{\bw}){\left(\mathscr{C}_W(\tau;\cdot)\right)}^\top.\label{eq:Gauss_psibWApp_sigm}
\end{align}\label{eq:Gauss_psibWApp}%
\end{subequations}%
From Corollary~\ref{corr:FSRPD_def_DPV} and \eqref{eq:Gauss_psibWApp}, the FSRPD for a \DPV{} with a Gaussian disturbance $\bw\sim \mathcal{N}(\overline{\mu}_{\bw}, \Sigma_{\bw})$ is,
\begin{align}
    \psibxzNotail&\sim \mathcal{N}( \overline{\mu}_{\bx_\tau}, \Sigma_{\bx_\tau}) \label{eq:Gauss_psibx_pdf}
\end{align}
with $\overline{\mu}_{\bx_\tau}=\xunpertNoTail + \overline{\mu}_{\bW_\tau}\in \mathbb{R}^n$ and $\Sigma_{\bx_\tau}=\Sigma_{\bW_\tau}$.
A recursive estimation of $\overline{\mu}_{\bx_\tau}$ and $\Sigma_{\bx_\tau}$ is also available via Kalman filters (see~\cite[Sec. III.A]{blackmore_probabilistic_2006}).
Note that Corollary~\ref{corr:FSRPD_def_DPV} requires only the definition of $\psibWtauNoTail$, which is the density of a random vector under linear transformation.
For non-Gaussian disturbances, we can use Fourier transforms (see \eqref{eq:cfun_def}--\eqref{eq:cfun_ift}) to compute the PDF $\psibWtauNoTail$ or the probability $\ProbbWtau\{\bW_\tau \in \mathcal{G}\}$ for some $ \mathcal{G}\in \sigmaAlg( \mathcal{W}^N)$.

Next, we obtain an exact description of the support of $\bx_\tau$, the forward stochastic reach set.
Note that Theorem~\ref{thm:FSRset_oplus} provides a tighter representation than in~\cite[Lem. 3]{VinodHSCC2017}. 
\begin{thm}{\textbf{(FSR set for \DPV{})}}\label{thm:FSRset_oplus} 
    For $\tau\in \mathbb{N}_{[1,N]}$, $$\FSRsetDPVNoTail= \{ \xunpertNoTail\} \oplus \mathscr{C}_{W}(\tau;\cdot) \mathcal{W}^\tau.$$
\end{thm}
\begin{proof}
    Since $\FSRsetDPVNoTail=\mathrm{supp}( \bx_\tau)$, we have
    \begin{align}
        \FSRsetDPVNoTail &=\left\{ \overline{z}\in \mathcal{X}: \forall r>0,\ \ProbbxNoTail\{\bx_\tau\in \Ball( \overline{z},r)\}>0\right\} \nonumber \\
        &=\left\{ \overline{z}\in \mathcal{X}: \forall r>0,\ \ProbbWtau\{\bW_\tau\in \Ball( \overline{z}-\xunpertNoTail,r)\}>0\right\} \nonumber\\
        &=\left\{ \overline{y}=\overline{z}-\xunpertNoTail: \forall r>0,\ \ProbbWtau\{\bW_\tau\in \Ball( \overline{y},r)\}>0\right\}. \nonumber
    \end{align}
    Hence, $\FSRsetDPVNoTail= \mathrm{supp}( \bW_\tau) \oplus \{\xunpertNoTail\}$.
    The proof is complete with the observation from \eqref{eq:ProbbW_defn} that $ \mathrm{supp}(\bW_\tau) = \mathscr{C}_W(\tau;\cdot) \mathcal{W}^\tau$.
\end{proof}
\begin{rem}
    For any $ \mathcal{G}\in \sigmaAlg( \mathcal{X})$, $ \mathcal{G}\cap \FSRsetDPVNoTail =\emptyset$ implies $\ProbbxNoTail\{\bx_\tau\in \mathcal{G}\}=0$ by \eqref{eq:supp_defn}.
\end{rem}

Theorems~\ref{thm:FSRPM_def_DPV} and~\ref{thm:FSRset_oplus} solve Problem~\ref{prob_st:FSR_DPV}.
We define the dimension of the FSR set as the dimension of its affine hull~\cite[Sec. 2.1.3]{boyd_convex_2004}.
By Theorem~\ref{thm:FSRset_oplus}, we see that the affine hull of the FSR set is contained in the column space\footnote{Range of the linear transformation $\mathscr{C}_W(\tau;\cdot)$\cite[Sec. 2.4]{strang_linear_2006}.} of $ \mathscr{C}_W(\tau;\cdot)$, a subspace of $ \mathcal{X}$.
Clearly, the FSR set is not full-dimensional when column rank of $ \mathscr{C}_W(\tau;\cdot)$ is not $n$.
In such cases, $\bx_t$ is an affine transformation of a lower-dimensional random vector which has a full-dimensional support.

\begin{figure}
    \centering
    \includegraphics[width=0.4\linewidth,Trim=\trimValuesConfEll,clip]{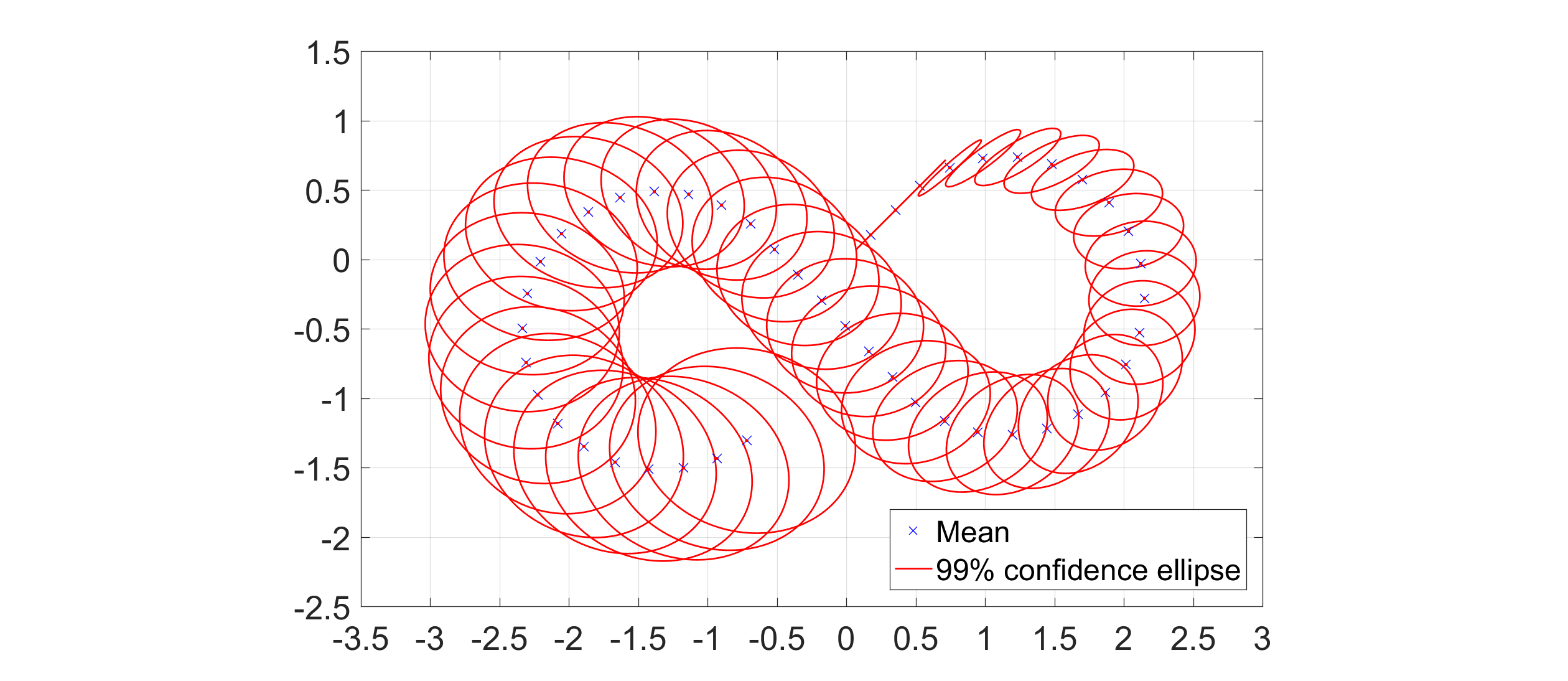} 
    \caption{Confidence region for the center of mass of an obstacle with unicycle dynamics \eqref{eq:obs_DMSP} starting at the origin with a predetermined turning rate sequence. 
See Appendix~\ref{app:UnicycleVals} for numerical values. The computation took $8.7$ ms.}
    \label{fig:unicycle_FSRPD}
\end{figure}
Figure~\ref{fig:unicycle_FSRPD} shows the $99\%$-confidence region ($0.01$-superlevel set of the PDF) for the center of mass of an obstacle with unicycle dynamics \eqref{eq:obs_DPV} for a predetermined turning rate sequence. 
The closed-form expressions for the FSRPD (Gaussian) of the corresponding \DPV{} is given in \eqref{eq:Gauss_psibx_pdf}. 
For any $\kappa\in \mathbb{R}, \kappa>0$, the $\kappa$-superlevel set of $\psibxzNotail$ is an ellipsoid $\mathcal{E}(\overline{\mu}_{\bx_\tau},Q_{\bx_\tau}(\kappa))$~\cite[Defn. 2.1.4]{KurzhanskiTextbook},
\begin{align}
    \mathcal{E}(\overline{\mu}_{\bx_\tau},Q_{\bx_\tau}(\kappa))&\triangleq\left\{ \overline{z}\in \mathcal{X}: {(\overline{z}- \overline{\mu}_{\bx_\tau})}^\top {(Q_{\bx_\tau}(\kappa))}^{-1}(\overline{z}- \overline{\mu}_{\bx_\tau})\leq 1 \right\}\label{eq:ell_defn}
\end{align}
with $Q_{\bx_\tau}(\kappa) = - 2\log{\left(\kappa\sqrt{\vert 2\pi\Sigma_{\bx_\tau}\vert}\right)}\Sigma_{\bx_\tau}$.
The initial turning rates are $\omega_0=\omega_1=0$ and $\omega_2\neq0$ (see Appendix~\ref{app:UnicycleVals}).
From \eqref{eq:obs_unicycle_dyn_param_DPV}, the heading $\theta_t=\frac{\pi}{4}$ for $t\in\{0,1,2\}$ and $\theta_3\neq\frac{\pi}{4}$.
This results in $ \mathscr{C}_W(t;\cdot)$ having rank $1$ for $t\in\{1,2,3\}$ and $ \mathscr{C}_W(4;\cdot)$ has rank $2$ (full rank) since $\theta_3\neq \theta_t$ for $t\in\{1,2,3\}$.
Specifically, the FSR set/the confidence region collapses to a line/line segment along the sole linearly independent column in $ \mathscr{C}_W(t;\cdot)$, the initial heading direction, for $t\in\{1,2,3\}$.

\subsection{Convexity properties of the forward stochastic reachability of \DPV{} dynamics}
\label{sub:FSR_cvx}

Understanding the convexity properties of the FSRPM and the FSRPD is useful for tractability.
Many standard distributions are log-concave, including Gaussian, uniform, and exponential distributions~\cite[e.g. 3.39, 3.40]{boyd_convex_2004}.

Lemma~\ref{lem:FSRPM_log_concave} follows from Theorem~\ref{thm:FSRPM_def_DPV} and~\cite[Lem. 2.1]{dharmadhikari1988unimodality}, and Lemma~\ref{lem:other_cvx} follows from~\cite[Thms 2.5 and 2.8]{dharmadhikari1988unimodality}.
\begin{lem}{\textbf{(Log-concave FSRPM for \DPV{})}}\label{lem:FSRPM_log_concave}
    If $\Prob_{\bw}$ is a log-concave probability measure, then the FSRPM $\ProbbxNoTail$ is log-concave over $ \mathcal{X}$ for $\forall\tau\in \mathbb{N}_{[1,N]}$.
\end{lem}
\begin{lem}{\textbf{(Convex FSR set and log-concave FSRPD for \DPV{})}}\label{lem:other_cvx}
    {\renewcommand{\theenumi}{\alph{enumi}}
    \begin{enumerate}
        \item Log-concave FSRPM  $\ProbbxNoTail$ $ \Rightarrow$ convex $\FSRsetDPVNoTail$.\label{prop:FSRset_cvx_partnumber}
        \item Log-concave FSRPM $\ProbbxNoTail$ and a full-dimensional FSR set $\FSRsetDPVNoTail$ $ \Leftrightarrow$ log-concave FSRPD $\psibxzNotail$.\label{prop:FSRPD_log_concave_partnumber}
    \end{enumerate}
    }
\end{lem}

For the example discussed in Figure~\ref{fig:unicycle_FSRPD}, we saw that the FSR set is not full-dimensional for $t\in\{1,2,3\}$.
By Lemma~\ref{lem:FSRPM_log_concave}, $\bx_t$ has a log-concave probability measure for every $t$, and a log-concave PDF over $\mathbb{R}^2$ for $t\geq 4$. 
By Lemma~\ref{lem:other_cvx}b, for $t\in\{1,2,3\}$, the random variable corresponding to the restriction of $\bx_t$ to the line along the initial heading direction has a log-concave PDF over $\mathbb{R}$.


\subsection{Forward stochastic reachability for \DMSP{} dynamics}
\label{sub:FSR_DMSP}

Consider a \DMSP{} \eqref{eq:DMSP_sys} initialized to  $ \overline{s}_0$, $ \overline{\lambda}_0$, and $ \UN$.
Given a Borel set $ \mathcal{G}_S\in \sigmaAlg( \mathcal{S})$, we define a set $\QtauProj$ that consists of the discrete state realizations $ \overline{q}_{\tau}$ such that the event $ \mathcal{G_S}$ has a non-zero probability of occurrence at time $\tau$.
Formally,
\begin{align}
    \QtauProj&=\{ \overline{q}_\tau \in \mathcal{Q}^\tau: q_\tau \in \Projq( \mathcal{G}_S),\  \ProbOvbsigTauMOne\{ \overline{\bsig}_{\tau-1} = \overline{q}_{\tau-1}\}>0\}.\label{eq:QtauProj}
\end{align}
using $\Projq$ defined in \eqref{eq:ProjectionDefn}.

Forward stochastic reachability of \eqref{eq:DMSP_sys} initialized to $ \overline{s}_0$, $ \overline{\lambda}_0$, and $ \UN$ consists of the FSRPM $\Probbfs$ and the FSR set $\FSRsetDMSP$, that is the probability measure and the support of the state of the \DMSP{} $\bfs_\tau$. 
We omit the definition of a FSRPD for a \DMSP{} \eqref{eq:DMSP_sys} since the state $\bfs_\tau$ is a mixed random vector~\cite[Sec. 5.3]{GubnerProbability2006}.

\begin{thm}{\textbf{(FSRPM for \DMSP{})}} \label{thm:FSRPM_def_DMSP}
    For any time instant $\tau\in\mathbb{N}_{[1,N]}$, the FSRPM $\Probbfs$ is given by 
    \begin{align}
        \Probbfs\{\bfs_\tau\in \mathcal{G}_S\}&=\sum_{\overline{q}_{\tau}\in \QtauProj}\Big(\Probbx\left\{\bx_\tau\in\mathcal{G}_X(q_\tau)\right\} \times\ProbOvbsigTau\{\overline{\bsig}_{\tau}=\overline{q}_{\tau}\}\Big).\label{eq:JM_prob_1}
    \end{align}
\end{thm}
\begin{proof}
    We simplify $\Probbfs\{\bfs_\tau\in \mathcal{G}_S\}$ by Lemma~\ref{lem:BorelS}.
    \begin{align}
        \Probbfs\{\bfs_\tau\in \mathcal{G}_S\} &= \Probbfs\left\{\bfs_\tau\in \bigcup_{q_\tau\in \Projq( \mathcal{G}_S)} \left(\{q_\tau\}\times \mathcal{G}_X(q_\tau)\right)\right\} \nonumber \\
                &= \sum_{q_\tau\in \Projq( \mathcal{G}_S)} \Probbfs\Big\{\bfs_\tau\in
    \{q_\tau\}\times \mathcal{G}_X(q_\tau)\Big\}.\nonumber
    \end{align}
    Next, we apply the law of total probability and use \eqref{eq:QtauProj},
    \begin{align}
    \Probbfs\{\bfs_\tau\in \mathcal{G}_S\} &= \sum_{\overline{q}_\tau\in \QtauProj}\Big(\ProbOvbsigTauMOne\{\overline{\bsig}_{\tau-1}=\overline{q}_{\tau-1}\}\nonumber\\
    &\quad\times\Probbfs\Big\{\bx_\tau\in\mathcal{G}_X(q_\tau),\bsig_\tau=q_\tau \Big\vert \overline{\bsig}_{\tau-1}=\overline{q}_{\tau-1}\Big\}\Big). \label{eq:FSRPM_DPV_intermediate_1}
    \end{align}
    From \eqref{eq:DMSP_sys}, $\bsig_\tau$ and $\bx_\tau$ are mutually independent for a fixed sequence of discrete state evolution $ \overline{q}_{\tau-1}$. 
    Hence, the second multiplicand in \eqref{eq:FSRPM_DPV_intermediate_1} simplifies to
    \begin{align}
        \Probbfs\Big\{\bx_\tau\in\mathcal{G}_X(q_\tau),\bsig_\tau=q_\tau \Big\vert \overline{\bsig}_{\tau-1}=\overline{q}_{\tau-1}\Big\} &= \Probbfs\{ \bsig_\tau=q_\tau\vert \overline{\bsig}_{\tau-1}=\overline{q}_{\tau-1}\}\nonumber \\
        &\qquad \times\Probbfs\left\{\bx_\tau\in\mathcal{G}_X(q_\tau)\vert \overline{\bsig}_{\tau-1}=\overline{q}_{\tau-1}\right\}. \label{eq:FSRPM_DPV_intermediate}
    \end{align}
    The first multiplicand  in \eqref{eq:FSRPM_DPV_intermediate} deals with the conditional probabilities associated solely with the discrete state $\bsig_t$, whose evolution is uninfluenced by $\bx_t$ (see \eqref{eq:DMSP_sys} and \eqref{eq:ProbBarSig}). 
    Hence, we simplify the first multiplicand in \eqref{eq:FSRPM_DPV_intermediate}, and combine it with the first multiplicand of \eqref{eq:FSRPM_DPV_intermediate_1},
    \begin{align}
        \Probbfs\{ \bsig_\tau=q_\tau\vert \overline{\bsig}_{\tau-1}=\overline{q}_{\tau-1}\} &=\ProbOvbsigTau\{ \bsig_\tau=q_\tau\vert \overline{\bsig}_{\tau-1} =\overline{q}_{\tau-1}\} \nonumber \\
        \ProbOvbsigTau\{ \bsig_\tau=q_\tau\vert \overline{\bsig}_{\tau-1} =\overline{q}_{\tau-1}\}\ProbOvbsigTauMOne\{\overline{\bsig}_{\tau-1}=\overline{q}_{\tau-1}\} &=\ProbOvbsigTau\{\overline{\bsig}_{\tau}=\overline{q}_{\tau}\}.\nonumber
    \end{align} 
    The second multiplicand in \eqref{eq:FSRPM_DPV_intermediate} is the FSRPM of the \DPV{} with the parameter sequence $\overline{\Lambda}_{\tau-1}= \LlambdatauMinusOne(\overline{q}_{\tau-1})$.
    This completes the proof.
\end{proof}

The number of summands in \eqref{eq:JM_prob_1}, $\vert \QtauProj \vert$, grows exponentially with $\tau$.
However, the growth rate depends on $ \mathscr{P}_{\bsig}$.
For the switching law in \eqref{eq:MarkovSwitchLaw}, we see that increasing $\tausw$ (infrequent switching) substantially reduces the number of summands.
The forward stochastic reachability of a \DMSP{} is given by a mixture probability which rules out simple sufficient conditions for the log-concavity of the FSRPM like Lemma~\ref{lem:FSRPM_log_concave}.
\begin{thm}{\textbf{(FSR set for \DMSP{})}}\label{thm:FSRset_oplus_DMSP} 
    For $\tau\in \mathbb{N}_{[1,N]}$, 
    \begin{align}
        \FSRsetDMSP &=\bigcup_{\overline{q}_{\tau}\in \QtauProj}\Big( q_\tau, \FSRsetDPVLlambda\Big).\nonumber
    \end{align}
\end{thm}
Theorem~\ref{thm:FSRset_oplus_DMSP} follows from Theorems~\ref{thm:FSRset_oplus} and~\ref{thm:FSRPM_def_DMSP}, by the enumeration of all the probable discrete state sequences $\overline{q}_\tau$ and the supports of the continuous state of the corresponding \DPV{} dynamics.
Theorems~\ref{thm:FSRPM_def_DMSP} and~\ref{thm:FSRset_oplus_DMSP} addresses Problem~\ref{prob_st:FSR_DMSP}.

\section{\Occupfunstext{}: theory and computation}
\label{sec:occup}

We will now apply the developed forward stochastic reachability methods to define a probabilistic occupancy functions and the \avoidsettext{} for rigid body obstacles that have \DPV{}/\DMSP{} dynamics.
Tables~\ref{tab:prop_occup} and~\ref{tab:prop_AS} summarize the results for the \occupfuntext{} and the \avoidsetstext{} discussed in this section.
\begin{table}
    \centering
    \begin{tabular}{|c|c|c|}
    \cline{2-3}
    \multicolumn{1}{c|}{} & \multicolumn{2}{c|}{\Occupfunstext{}} \\\cline{2-3}
    \multicolumn{1}{c|}{} &  $\occupDPVNoTail$ & $\occupDMSPNoTail$ \\ \hline               
    Obstacle dynamics        & \DPV{} & \DMSP{}  \\ \hline
    Definition               & \eqref{eq:occup_DPV_defn}--\eqref{eq:occup_DPV_as_conv}, Proposition~\ref{prop:OccupDPV_symm} & \eqref{eq:occup_DMSP} \\ \hline
    Log-concavity              & Proposition~\ref{prop:OccupDPV_logcon} &  - \\ \hline
    Upper semicontinuous     & Proposition~\ref{prop:OccupDPV_usc}    & Proposition~\ref{prop:OccupDMSP_prop}\ref{prop:OccupDMSP_prop_usc}    \\ \hline
   \end{tabular}
   \caption{Properties of the \occupfuntext{} discussed in Section~\ref{sec:occup}}
   \label{tab:prop_occup}
\end{table}
\begin{table}
    \centering
    \tabcolsep=4pt
    \begin{tabular}{|c|c|c|}
    \cline{2-3}
    \multicolumn{1}{c|}{} & \multicolumn{2}{c|}{\avoidsettext{}} \\\cline{2-3}
    \multicolumn{1}{c|}{} & $\ASDPVNoTail$ & $\ASDMSPNoTail$ \\ \hline               
    Obstacle dynamics & \DPV{} & \DMSP{}  \\ \hline
    Definition        & \eqref{eq:ASDPV_defn} & \eqref{eq:ASDMSP_defn} \\ \hline
    Convexity         & Proposition~\ref{prop:OccupDPV_logcon} & - \\ \hline
    Closedness        & Proposition~\ref{prop:OccupDPV_usc}    & Proposition~\ref{prop:OccupDMSP_prop}\ref{prop:OccupDMSP_prop_usc} \\ \hline
    Boundedness       & Proposition~\ref{prop:ASDPV_bounded}   & Proposition~\ref{prop:OccupDMSP_prop}\ref{prop:OccupDMSP_prop_bounded} \\ \hline
    Compactness       & Theorem~\ref{thm:convexCompactDPV}  & Proposition~\ref{prop:OccupDMSP_prop}\ref{prop:OccupDMSP_prop_compact} \\ \hline
    Cover of compact                 & \multirow{2}{*}{\begin{minipage}{1cm}\centering - \end{minipage}}  & \multirow{2}{*}{\begin{minipage}{2cm}\centering Theorem~\ref{thm:AvoidSetUnionDMSP} \end{minipage}} \\ 
    and convex sets                 &   &  \\ \hline
   \end{tabular}
   \caption{Properties of the \avoidsettext{} discussed in Section~\ref{sec:occup}}
   \label{tab:prop_AS}
\end{table}

To define collision, we consider a rigid body robot with shape $ \mathcal{R}( \overline{0})\subseteq \mathcal{X}$ which also satisfies Assumptions~\ref{assum:Borel} and~\ref{assum:obs}.
From~\cite[Sec. 4.3.2]{lavalle2006planning}, we see that the probability of collision between the rigid body obstacle and the rigid body robot is equivalent to the probability of collision between the robot with exactly same dynamics but shape reduced to a point and the obstacle with an augmented shape $ \mathcal{O}( \overline{0}) \oplus (-\mathcal{R}(0))$.  
This motivates the definition of collision probability using the probability with which the obstacle ``occupies'' the state of the point robot.
Without loss of generality, we assume the robot shape is a point and $\mathcal{O}( \overline{0})$ is the appropriately augmented obstacle shape in the sequel.

\subsection{\Occupfunstext{} and \avoidsetstext{} for a rigid body obstacle with \DPV{} dynamics}
\label{sub:occup_DPV}

Consider an obstacle with \DPV{} dynamics \eqref{eq:DPV_sys} initialized to $ \overline{x}_0, \overline{\lambda}_0$, $\overline{\Lambda}_{N-1}$, and $\UN$, and a shape $ \mathcal{O}( \overline{0})$ defined by $h_\mathrm{obs}: \mathcal{X} \rightarrow \mathbb{R}$ (see \eqref{eq:obs_rigidbody_defn}). 
We can define \occupfuntext{} $\occupDPV :\mathcal{X} \rightarrow [0,1]$ using Lemma~\ref{lem:rigidBody}\ref{lem:rigidBody1} as,
\begin{align}
     \occupDPVNoTail&=\ProbbxNoTail\left\{\bx_\tau\in\left\{\overline{z}\in \mathcal{X}:\overline{y}\in \mathcal{O}(\overline{z})\right\}\right\} \label{eq:occup_DPV_defn} \\
              &=\ProbbxNoTail\left\{\bx_\tau\in(-\mathcal{O}(-\overline{y}))\right\}.\label{eq:occup_DPV_as_prob}
\end{align}
For brevity, we will omit references to $ \overline{\lambda}_0$, $\overline{x}_0$, $\overline{\Lambda}_{N-1}$, $\UN$, and $h_\mathrm{obs}$.
Separately from \eqref{eq:occup_DPV_defn} and \eqref{eq:occup_DPV_as_prob}, $\occupDPVnothing$ can also be defined:
\begin{enumerate}
    \item using Lemma~\ref{lem:rigidBody}\ref{lem:rigidBodyCenter},
        \begin{align}
            \occupDPVNoTail=\ProbbxNoTail\left\{\bx_\tau\in-\left(\{-\overline{y}\}\oplus-\mathcal{O}(\overline{0})\right)\right\} &=\ProbbxNoTail\left\{\bx_\tau\in\left(\{\overline{y}\}\oplus-\mathcal{O}(\overline{0})\right)\right\} \label{eq:occup_DPV_as_prob_separated_positive} \\
                            &=\ProbbxNoTail\left\{(\overline{y}-\bx_\tau)\in\mathcal{O}(\overline{0})\right\} \label{eq:occup_DPV_as_prob_separated}
        \end{align}
    \item using expectations ($\ProbbxNoTail\{\bx_\tau\in \mathcal{G}_X\}=\ExpbxNoTail\left[1_{ \mathcal{G}_X}( \bx_\tau)\right]$),
        \begin{align}
            \occupDPVNoTail&=\ExpbxNoTail\left[1_{\mathcal{O}(\overline{0})}(\overline{y}-\bx_\tau)\right]\label{eq:occup_DPV_as_exp}
        \end{align}
    \item using convolution (see~\cite{HomChaudhuriACC2017}): By \eqref{eq:occup_DPV_as_exp} and Lemma~\ref{lem:rigidBody}\ref{lem:rigidBody2},
        \begin{align}
            \occupDPVNoTail &=\int_{ \mathcal{X}} \psibxzNotail 1_{\mathcal{O}(\overline{0})}(\overline{y}-\overline{z})d\overline{z} \label{eq:occup_DPV_intg} \\
                             &=[\psibxnoarg\ast 1_{\mathcal{O}(\overline{0})}](\overline{y}).\label{eq:occup_DPV_as_conv}
        \end{align}
\end{enumerate}

We obtain a simpler description of $\occupDPVnothing$ for centrally symmetric obstacles, which
follows from \eqref{eq:obs_rigidbody_defn_reflect} and \eqref{eq:occup_DPV_as_prob_separated}.
\begin{prop}{\textbf{($\occupDPVnothing$ under central symmetry)}}\label{prop:OccupDPV_symm}
    For a centrally symmetric rigid body obstacle, the \occupfuntext{} is given by $\occupDPVNoTail=\ProbbxNoTail\left\{\bx_\tau\in \mathcal{O}(\overline{y})\right\}$.
\end{prop}
Evaluating $\occupDPVNoTail$ over a grid on $ \mathcal{X}$ provides the occupancy grid as defined in~\cite{elfes1989using}. 
Occupancy grids have been used in mapping and navigation~\cite{elfes1989using, thrun_probabilistic_2005}.
However, improving the accuracy in grid-based approach typically comes at significantly high computational cost.
With $\occupDPVNoTail$, we have a grid-free formulation to characterize ``occupancy'' by a stochastic rigid body obstacle.
Also, the definition of collision probability given in~\cite[eq. (8)]{du2011probabilistic} coincides with $\phi_{\bx}$ under Proposition~\ref{prop:OccupDPV_symm}, since the robot state in our formulation is deterministic. 

Given $\alpha \in \mathbb{R},\ \alpha\geq 0$, the \avoidsettext{} is
\begin{align}
    \ASDPV =\{\overline{y}\in \mathcal{X} : \occupDPV \geq \alpha\}\label{eq:ASDPV_defn}
\end{align}
where the subscript $P$ denotes that the obstacle dynamics are \DPV{}.
From \eqref{eq:ASDPV_defn}, we have $\forall \tau\in \mathbb{N}_{[1,N]}$ and $\forall\alpha'\in \mathbb{R},\ \alpha\leq\alpha'$,
\begin{subequations}
    \begin{align}
            \ASDPVNoTail&\supseteq \ASDPVNoTailPrime.\label{eq:ASDPV_subset}\\
    \ASDPVNoTail &=\begin{cases}
    {\arraycolsep=0pt
        \begin{array}{ll}
            \mathcal{X}, &\  \alpha= 0 \\
            \emptyset, &\  \alpha>\occupDPVymax \\
        \end{array}
    }
    \end{cases} \label{eq:ASDPV_full_empty}%
\end{align}\label{eq:ASDPV_properties}%
\end{subequations}%
where $\ymax\in \mathcal{X}$ is the maximizer of $\occupDPVNoTail$.
Equation \eqref{eq:ASDPV_full_empty} follows from the fact that $\forall\tau\in \mathbb{N}_{[1,N]}$, $\occupDPVNoTail$ is non-negative, and $\occupDPVymax\geq \occupDPVNoTail, \forall \overline{y}\in \mathcal{X}$.
By definition, $\occupDPVymax\leq 1$.

\begin{prop}{\textbf{(Log-concave $\occupDPVnothing$ and convex $\ASDPVnothing$)}}
    If $\Prob_{\bw}$ is log-concave over $ \mathcal{W}$ and $ \mathcal{O}(\overline{0})$ is convex, then $\occupDPVNoTail$ is log-concave over $ \overline{y}$ $\forall\tau\in \mathbb{N}_{[1,N]}$. Moreover, $\ASDPVNoTail$ is convex $\forall\alpha\in \mathbb{R}$.~\label{prop:OccupDPV_logcon}
\end{prop}
\begin{proof}
    From Lemma~\ref{lem:FSRPM_log_concave}, we know the FSRPM $\ProbbxNoTail$ is log-concave under these conditions.
    Note that the set $-\mathcal{O}(\overline{0})$ is convex since $\mathcal{O}(\overline{0})$ is convex and convexity is preserved under linear transformation (see \eqref{eq:obs_rigidbody_defn_reflect} and~\cite[Sec. 2.3.2]{boyd_convex_2004}).
    Using the definition \eqref{eq:occup_DPV_as_prob_separated_positive} and property \eqref{eq:log-concaveProbSetMove}, we conclude that $\occupDPVNoTail$ is log-concave over $ \overline{y}\in \mathcal{X}$  $\forall\tau\in \mathbb{N}_{[1,N]}$.

    The set $\ASDPVNoTail$ is convex since log-concave functions are quasiconcave~\cite[Sec. 3.5]{boyd_convex_2004}.
\end{proof}
Proposition~\ref{prop:OccupDPV_logcon} provides the conditions under which the computation of $\ymax$ may be posed as an unconstrained log-concave optimization problem, which may be tractably solved when $\phi_{\bx}$ is given.
Alternatively, we can avoid the computation of $\ymax$ completely by using Proposition~\ref{prop:ymax} for centrally symmetric disturbances, like a zero-mean Gaussian disturbance.
\begin{lem}
    If $ \Prob_{\bw}$ is centrally symmetric, then $\ProbbWtau$ is centrally symmetric. \label{lem:symmetric_ProbbWtau}
\end{lem}
\begin{proof}
    See Appendix~\ref{app:proof_symmetric_ProbbWtau}. 
\end{proof}
\begin{prop}{ \textbf{($\ymax$ under symmetry)}}\label{prop:ymax}
    If $ \Prob_{\bw}$ is centrally symmetric and log-concave, and $ \mathcal{O}(\overline{0})$ is centrally symmetric and convex, then $\ymax = \xunpertNoTail$.
\end{prop}
\begin{proof}
    For every $\overline{y}\in \mathcal{X}$, we define $\yshift = \overline{y}-\xunpertNoTail \in \mathcal{X}$, and $\ell: \mathcal{X} \rightarrow [0,1]$,
    \begin{align}
        \ell(\yshift)\triangleq\occupDPVshift=\occupDPVNoTail.
    \end{align}
    We will show that I) $\ell(\yshift)$ is even in $\yshift\in \mathcal{X}$, II) $\ell(\yshift)$ is log-concave in $\yshift\in \mathcal{X}$, and III) $ \overline{0}$ is the maximizer of $\ell$ since it is log-concave and even.
    Using III), $\ell( \overline{0})\geq \ell(\yshift) \equiv \occupDPVxunpert\geq \occupDPVNoTail$ for all $\yshift\in \mathcal{X}$ and the corresponding $ \overline{y}$, which completes the proof.

    \emph{Part I) $\ell$ is even over $\yshift$}: Since $ \mathcal{O}( \overline{0})$ is centrally symmetric, we use Proposition~\ref{prop:OccupDPV_symm} to define $\occupDPVNoTail$.
    By definition of $\ell(\cdot)$,
    \begin{align}
        \ell(\yshift)&=\ProbbxNoTail\{\bx_\tau \in \mathcal{O}( \xunpertNoTail+\yshift)\} &=\ProbbxNoTail\{\bx_\tau \in \mathcal{O}( \overline{0}) \oplus \{ \xunpertNoTail+\yshift\}\} \label{eq:ymax_symm1} \\
                  &=\ProbbWtau\{\bW_\tau \in \mathcal{O}( \overline{0}) \oplus \{\yshift\}\}. \label{eq:ymax_symm2}
    \end{align}
    where \eqref{eq:ymax_symm1} follows from Lemma~\ref{lem:rigidBody}\ref{lem:rigidBodyCenter}, and \eqref{eq:ymax_symm2} follows from Theorem~\ref{thm:FSRPM_def_DPV}.
    By Lemmas~\ref{lem:rigidBody}\ref{lem:rigidBodyCenter} and~\ref{lem:symmetric_ProbbWtau}, we have $\ell(\yshift)=\ell(-\yshift)$.

    \emph{Part II) $\ell$ is log-concave over $\yshift$}: From Proposition~\ref{prop:OccupDPV_logcon}, we know that $\log{(\occupDPVNoTail)}$ is concave in $ \overline{y}$. Since composition of a concave function with an affine function preserves concavity~\cite[Sec. 3.2.2]{boyd_convex_2004}, $\ell(\yshift)$ is log-concave.
    
    \emph{Part III) $ \overline{0}$ is the maximizer of $\ell$}: By definition of log-concavity, $\ell( \overline{0})\geq \sqrt{\ell(\yshift)\ell(-\yshift)},\ \forall\yshift\in \mathcal{X}$~\cite[Sec. 3.5.1]{boyd_convex_2004}. 
    Since $\ell$ is non-negative and $\ell(\yshift)$ is even in $\yshift$, we have $\ell( \overline{0})\geq \ell(\yshift),\ \forall\yshift\in \mathcal{X}$.
\end{proof}
\begin{prop}{\textbf{(U.s.c. $\occupDPVnothing$ and closed $\ASDPVnothing$)}}\label{prop:OccupDPV_usc}
    If $ \mathcal{O}(\overline{0})$ is closed, then $\forall\tau\in \mathbb{N}_{[1,N]}$,  $\occupDPVNoTail$ is upper semicontinuous, and $\ASDPVNoTail$ is closed $\forall\alpha\in \mathbb{R}$.
\end{prop}
\begin{proof}
    For every sequence $ \overline{y}_i \rightarrow \overline{y}$, we need to show that $\limsup_{i \rightarrow \infty} \occupDPVNoTailSeq \leq \occupDPVNoTail$.
    Using \eqref{eq:occup_DPV_as_exp}, this is equivalent to proving $\limsup_{i \rightarrow \infty} \ExpbxNoTail\left[1_{\mathcal{O}(\overline{0})}(\overline{y}_i-\bx_\tau)\right] \leq \ExpbxNoTail\left[1_{\mathcal{O}(\overline{0})}(\overline{y}-\bx_\tau)\right]$. 

    Since $ \mathcal{O}( \overline{0})$ is closed, the function $1_{\mathcal{O}( \overline{0})}( \overline{y})$ is u.s.c. in $ \overline{y}$.
    The function $-1_{\mathcal{O}(\overline{0})}(\overline{y}- \overline{z})$ is l.s.c in $ \overline{y}, \overline{z}$ since $-1_{\mathcal{O}( \overline{0})}( \overline{y})$ is l.s.c and $ \overline{y}- \overline{z}$ is continuous in $ \overline{y}, \overline{z}$~\cite[Ex. 1.4]{rockafellar_variational_2009}.
    Therefore, the function $f: \mathcal{X}^2 \rightarrow \{0,1\}$ with $f( \overline{y}, \overline{z}) = 1-1_{\mathcal{O}(\overline{0})}(\overline{y}- \overline{z})$ is l.s.c. in $ \overline{y}$ and $ \overline{z}$.
    Specifically, ${\liminf}_{i \rightarrow \infty} f(\overline{y}_i, \overline{z}) \geq  f(\overline{y}, \overline{z})$ for any $ \overline{z}$.
    Using~\cite[Thm. 1.12d]{RudinReal1987}, we also conclude that $f$ is Borel-measurable in $ \overline{y}$ and $\overline{z}$.
    By construction, $f(\overline{y}_i, \bx_\tau), f(\overline{y}, \bx_\tau)$ is non-negative (pointwise). 
    Hence, by Fatou's lemma~\cite[Sec. 6.2, Thm. 2.1]{ChowProbability1997}, the fact that $f$ is l.s.c, Borel-measurable, and non-negative, and linearity of the expectation operator, we have
    \begin{align}
        \liminf_{i \rightarrow \infty} \ExpbxNoTail\left[f(\overline{y}_i, \bx_\tau)\right] &\geq \ExpbxNoTail\left[\liminf_{i \rightarrow \infty} f(\overline{y}_i, \bx_\tau)\right] \geq \ExpbxNoTail\left[f(\overline{y}, \bx_\tau)\right] \nonumber\\
        \therefore \limsup_{i \rightarrow \infty} \ExpbxNoTail\left[1_{\mathcal{O}(\overline{0})}(\overline{y}_i-\bx_\tau)\right] &\leq \ExpbxNoTail\left[1_{\mathcal{O}(\overline{0})}(\overline{y}-\bx_\tau)\right].\nonumber
    \end{align}

    Closedness of $\ASDPVNoTail$ follows from the u.s.c. of $\occupDPVNoTail$ and \eqref{eq:ASDPV_defn}.
\end{proof}

\begin{prop}{\textbf{(Bounded $\ASDPVnothing$)}}
    If $ \mathcal{O}(\overline{0})$ is bounded, then $\ASDPVNoTail$ is bounded for every $\alpha>0$ and $\tau\in \mathbb{N}_{[1,N]}$.\label{prop:ASDPV_bounded}
\end{prop}
\begin{proof}
    Let $b, b_0\in \mathbb{R}$.
Consider some $ \overline{y} \in \mathcal{X}$ such that $ \overline{y}\in \ASDPVNoTail \Rightarrow \occupDPVNoTail > \alpha$. We need to show that for every unit vector $ \overline{d}\in \mathcal{X}$, there exists $b_0>0$ such that $\occupDPVNoTailbd <\alpha$ for all $b> b_0$.

Assume, for contradiction, that there is one unit vector $ \overline{d}_c\in \mathcal{X}$ for which no such $b_0$ exists, or equivalently, for every $b>0$, $\occupDPVNoTailbdc \geq \alpha$.

Since $ \mathcal{O}(\overline{0})$ is compact, there exists a ball $ \Ball( \overline{0},r)$ centered at origin with radius $r\in \mathbb{R}$, $r>0$ such that $-\mathcal{O}(\overline{0})\subset \Ball( \overline{0}, r)$.
From Lemma~\ref{lem:rigidBody}\ref{lem:rigidBodyCenter}, $ -\mathcal{O}(-(\overline{y}+b \overline{d}_c)) \subset \Ball( \overline{y}+b \overline{d}_c, r)$ for any $b>0$. 
Given $\alpha>0$, there exists $N_\alpha\in \mathbb{N}\setminus\{0\}$ such that $\alpha N_\alpha >1$ (Archimedean property~\cite[Corr. 5.4.13]{TaoAnalysisI}). 
We define $\GammaBi = \{iR: i\in \mathbb{N}_{[1, N_\alpha]} \}\subset \mathbb{R}$ where $ R > 2r$ which is a collection of $N_\alpha$ options for $b$.
By our choice of $R$, for any $b_1, b_2\in\GammaBi,\ b_1\neq b_2$, the sets
$\Ball( \overline{y}+b_1 \overline{d}_c;r)$ and $\Ball(\overline{y}+b_2 \overline{d}_c;r)$ are distinct, and thereby the sets $-\mathcal{O}(-(\overline{y}+b_1\overline{d}_c))$ and $-\mathcal{O}(-(\overline{y}+b_2\overline{d}_c))$ are distinct.
Construct a Borel set $\mathcal{G}_X = \cup_{b\in \GammaBi} (-\mathcal{O}(-(\overline{y}+b\overline{d}_c)))$, a union of mutually disjoint sets.
Hence, $\ProbbxNoTail\{\bx_\tau\in \mathcal{G}_X\}=\sum_{i\in\GammaBi}\ProbbxNoTail\{\bx_\tau\in (-\mathcal{O}(-(\overline{y}+b\overline{d}_c)))\}$.
From \eqref{eq:occup_DPV_as_prob}, our assumption on $ \overline{d}_c$, and the definition of $N_\alpha$, we have
\begin{align}
    \ProbbxNoTail\{\bx_\tau\in \mathcal{G}_X\}&= \sum_{i\in\GammaBi} \occupDPVNoTailbidc \geq N_\alpha\alpha>1 \nonumber
\end{align}
which leads to a contradiction, completing the proof.
\end{proof}

We utilize the Heine-Borel theorem and summarize the results from Propositions~\ref{prop:OccupDPV_logcon},~\ref{prop:OccupDPV_usc}, and~\ref{prop:ASDPV_bounded} as Theorem~\ref{thm:convexCompactDPV}.
Theorem~\ref{thm:convexCompactDPV} addresses Problem~\ref{prob_st:occupDPV_cvx_cmpt}.
\begin{thm}{\textbf{(Compact and convex $\ASDPVnothing$)}}~\label{thm:convexCompactDPV}
    If $\Prob_{\bw}$ is log-concave over $ \mathcal{W}$ and $ \mathcal{O}(\overline{0})$ is convex and compact, then $\ASDPVNoTail$ is a convex and compact set for all $\alpha>0$ and $\forall\tau\in \mathbb{N}_{[1,N]}$.
\end{thm}


\subsection{\Occupfunstext{} and \avoidsetstext{} for a rigid body obstacle with \DMSP{} dynamics}
\label{sub:occup_DMSP}

Consider a rigid body obstacle with \DMSP{} dynamics \eqref{eq:DMSP_sys} initialized to  $ \overline{s}_0 = {[q_0, \overline{x}_0^\top]}^\top, \overline{\lambda}_0$, and $\UN$, and shape $ \mathcal{O}( \overline{0})$ defined using an appropriate $h_\mathrm{obs}$. 
We define collision only using the continuous state, i.e., the robot and the obstacle are said to be in collision, if the state of the point robot $ \overline{y}$ and the obstacle $\bfs_\tau=[\bsig_\tau\  \bx_\tau^\top]^\top$ satisfies the condition,
\begin{align}
    \bfs_\tau\in\cup_{q\in \mathcal{Q}} (q, -\mathcal{O}(-\overline{y})) \quad \Leftrightarrow \quad \bx_\tau \in -\mathcal{O}(- \overline{y})\ \forall \bsig_\tau.\label{eq:occup_event_DMSP}
\end{align}
We collect all discrete state realizations $ \overline{q}_\tau$ that have a non-zero probability of occurrence under $ \mathscr{P}_{\bsig}$ \eqref{eq:sys_disc} in $\QtauMinus$. 
Formally, we define $\QtauMinus$ using \eqref{eq:QtauProj} as
\begin{align}
    \QtauMinus = \QtauProjObs\subseteq \mathcal{Q}^\tau.\label{eq:QtauMinusDefn}
\end{align}
The \occupfuntext{} of a rigid body obstacle with \DMSP{} dynamics, $\occupDMSP :\mathcal{X} \rightarrow [0,1]$, may then be defined using \eqref{eq:occup_event_DMSP}, Lemma~\ref{lem:rigidBody}\ref{lem:rigidBody1}, and Theorem~\ref{thm:FSRPM_def_DMSP} in terms of $\occupDPVnothing$,
\begin{align}
    \occupDMSP &= \ProbbfsNoU\left\{\bfs_\tau\in\cup_{q\in \mathcal{Q}} (q, -\mathcal{O}(-\overline{y}))\right\} \label{eq:occup_def_DMSP}\\
    &=\sum_{\overline{q}_{\tau}\in \QtauMinus}\Big(\ProbOvbsigTauNoTail\{\overline{\bsig}_{\tau}=\overline{q}_{\tau}\}\ProbbxDMSPAS\left\{\bx_\tau\in(-\mathcal{O}(-\overline{y}))\right\}\Big)\nonumber\\
    &=\sum_{\overline{q}_{\tau}\in \QtauMinus}\Big(\occupDPVLlambda\ProbOvbsigTauNoTail\{\overline{\bsig}_{\tau}=\overline{q}_{\tau}\}\Big). \label{eq:occup_DMSP}
\end{align}
Similarly to \eqref{eq:ASDPV_defn}, we define the $\alpha$-superlevel set of the \occupfuntext{} $\occupDMSPNoTail$ as the \avoidsettext{} which have properties identical to \eqref{eq:ASDPV_properties},
\begin{align}
    \ASDMSP =\{\overline{y}\in \mathcal{X} : \occupDMSP\geq \alpha\}\label{eq:ASDMSP_defn}.
\end{align}
Here, the subscript $S$ denotes that the obstacle dynamics are \DMSP{}.
For brevity, we will omit references to $ \overline{\lambda}_0$, $ \overline{s}_0$, $ \overline{x}_0$, $ \LlambdatauMinusOne(\cdot)$, $\Utau$, and $ h_\mathrm{obs}$.

\begin{prop}{\textbf{(U.s.c. $\occupDMSPnothing$ and closed, bounded, and compact $\ASDMSPnothing$)}}\label{prop:OccupDMSP_prop}
    {\renewcommand{\theenumi}{\alph{enumi}}
    \begin{enumerate}
        \item If $ \mathcal{O}(\overline{0})$ is bounded, then $\ASDMSPNoTail$ is bounded $\forall\alpha>0,\ \tau\in \mathbb{N}_{[1,N]}$.\label{prop:OccupDMSP_prop_bounded}
        \item If $ \mathcal{O}(\overline{0})$ is closed, then $\forall\tau\in \mathbb{N}_{[1,N]}$, $\occupDMSPNoTail$ is upper semicontinuous over $ \overline{y}\in\mathcal{X}$, and $\ASDMSPNoTail$ is closed $\forall\alpha\in \mathbb{R},\ \alpha\geq 0,\ \tau\in \mathbb{N}_{[1,N]}$.\label{prop:OccupDMSP_prop_usc}
        \item If $ \mathcal{O}(\overline{0})$ is compact, then $\ASDMSPNoTail$ is compact $\forall\alpha>0,\ \tau\in \mathbb{N}_{[1,N]}$.\label{prop:OccupDMSP_prop_compact}
    \end{enumerate}
    }
\end{prop}
\begin{proof}
    \emph{Proof of~\ref{prop:OccupDMSP_prop_bounded})}: Similar to the proof in Proposition~\ref{prop:ASDPV_bounded}.

    \emph{Proof of~\ref{prop:OccupDMSP_prop_usc})}: Follows from \eqref{eq:occup_DMSP}, Propositions~\ref{prop:OccupDPV_logcon} and~\ref{prop:OccupDPV_usc}, and the fact that non-negative finite sums of bounded u.s.c functions are u.s.c~\cite[Prop. B.3]{PuttermanMarkov2005}.

    \emph{Proof of~\ref{prop:OccupDMSP_prop_compact})}: By~\ref{prop:OccupDMSP_prop_bounded}),~\ref{prop:OccupDMSP_prop_usc}), and Heine-Borel theorem.
\end{proof}
We do not have sufficient conditions for log-concavity of $\occupDMSPnothing$ or convexity of $\ASDMSPnothing$ similar to Proposition~\ref{prop:OccupDPV_logcon}, since non-negative weighted sums of log-concave functions are not log-concave in general~\cite[Sec. 3.5]{boyd_convex_2004}.
\begin{thm}{\textbf{(Covering of $\ASDMSPnothing$ using convex and compact $\ASDPVnothing$)}}~\label{thm:AvoidSetUnionDMSP}
    If $\Prob_{\bw}$ is log-concave over $ \mathcal{W}$ and $ \mathcal{O}(\overline{0})$ is convex and compact, then $\ASDMSPNoTail$ is covered by a union of convex and compact sets for $\alpha>0$ and $\forall\tau\in \mathbb{N}_{[1,N]}$,
    \begin{align}
        \ASDMSPNoTail&\subseteq\bigcup_{ \overline{q}_\tau\in \QtauMinus} \ASDPValpha \label{eq:overapproxDMSP}
\end{align}
with $\alphaSqtau= \frac{\alpha}{\vert\QtauMinus\vert\ \ProbOvbsigTau\{\overline{\bsig}_{\tau}=\overline{q}_{\tau}\}}$.
\end{thm}
\begin{proof}
    By definition of $\QtauMinus$, $\alphaSqtau$ is well-defined.
    We will prove the contrapositive of \eqref{eq:overapproxDMSP}, i.e., for some $ \overline{y}\in \mathcal{X}$, $\overline{y}\not\in\bigcup_{ \overline{q}_\tau\in \QtauMinus} \ASDPValphaNoTail$ implies $\overline{y}\not\in\ASDMSPNoTail$.
    Existence of a $ \overline{y}\in \mathcal{X}$ such that $\overline{y}\not\in\bigcup_{ \overline{q}_\tau\in \QtauMinus} \ASDPValphaNoTail$ is guaranteed by Theorem~\ref{thm:convexCompactDPV}, $\alpha>0$, and the fact that a finite union of compact sets is compact~\cite[Thm. 12.5.10]{TaoAnalysisII}.

    From \eqref{eq:ASDPV_defn}, we have for every ${ \overline{q}_\tau\in \QtauMinus}$, 
    \begin{align}
        \occupDPVNoTailbeta&< \frac{\alpha}{\vert\QtauMinus\vert\ProbOvbsigTau\{\overline{\bsig}_{\tau}=\overline{q}_{\tau}\}} \nonumber \\
        \therefore\ProbOvbsigTau\{\overline{\bsig}_{\tau}=\overline{q}_{\tau}\}\occupDPVNoTailbeta&< \frac{\alpha}{\vert\QtauMinus\vert}. \label{eq:subset_proof}
\end{align}
From \eqref{eq:occup_DMSP} and \eqref{eq:subset_proof}, 
\begin{align}
        \occupDMSPNoTail = \sum_{ \overline{q}_\tau \in \QtauMinus}\ProbOvbsigTau\{\overline{\bsig}_{\tau}=\overline{q}_{\tau}\} \occupDPVNoTailbeta&< \alpha.\nonumber
\end{align}
By \eqref{eq:ASDMSP_defn}, $ \overline{y}\not\in\ASDMSPNoTail$.

Theorem~\ref{thm:convexCompactDPV} shows that the over-approximation \eqref{eq:overapproxDMSP} is a finite union of convex and compact sets.
\end{proof}

Theorem~\ref{thm:AvoidSetUnionDMSP} addresses Problem~\ref{prob_st:occupDMSP_algo}.
We can use a collection of $\ASDPVnothing$ corresponding to appropriately defined \DPV{} dynamics to overapproximate the (non-convex) keep-out region.

\begin{rem}
    In the case of multiple rigid body obstacles with \DMSP{}/\DPV{} dynamics, risk allocation-based techniques~\cite[Sec. 3.D]{blackmore_probabilistic_2006} replace the ``joint'' safety guarantee with a conservative collection of ``individual'' safety guarantees, resulting in a union of keep-out regions defined for individual obstacles.  
\end{rem}

\section{Computation of $\ASDPVnothing$}
\label{sec:compute}

In this section, we will address Problem~\ref{prob_st:occupDPV_algo} and compute approximations of $\ASDPVnothing$ for a rigid body obstacle with \DPV{} dynamics initialized to $ \overline{x}_0, \overline{\lambda}_0$, $\overline{\Lambda}_{N-1}$, and $\UN$.
By Theorem~\ref{thm:AvoidSetUnionDMSP}, we can use these sets to overapproximate $\ASDMSPnothing$ for a rigid body obstacle with \DMSP{} dynamics as well.

Unfortunately, an exact, closed-form expression for $\occupDPVnothing$ is typically unavailable.
To avoid calculating \eqref{eq:ASDPV_defn} via a grid over the state space $ \mathcal{X}$, which is computationally expensive and lacks scalability to higher dimensional systems, we propose two grid-free alternatives to compute $\ASDPVnothing$,
\begin{enumerate}
    \item projection-based tight polytopic overapproximation and underapproximation (Algorithm~\ref{algo:ASDPV_poly}), and 
    \item Minkowski sum-based overapproximation (Algorithm~\ref{algo:MinkSum}).
\end{enumerate}
In particular, we seek an overapproximation of $\occupDPVnothing$ so that we remain conservative with respect to safety, i.e., avoiding the overapproximations of $\occupDPVnothing$ still provides the desired safety guarantees.
An underapproximation of the keep-out regions, when available, provides insight on the degree of conservativeness.
These algorithms are recursion-free by the use of forward stochastic reachability via Fourier transforms.

We can also split the computation of $\ASDPVnothing$ into an offline and online computation, using Proposition~\ref{prop:linear_exp} which follows from \eqref{eq:traj_LTV}, \eqref{eq:occup_DPV_as_prob_separated}, and \eqref{eq:ASDPV_defn}.
\begin{prop}{\textbf{(Effect of $ \overline{x}_0$ on $\occupDPVnothing$ and $\ASDPVnothing$)}}~\label{prop:linear_exp}
    Consider the \DPV{} \eqref{eq:DPV_sys} initialized to $ \overline{\lambda}_0$, $\overline{x}_0$, $ \overline{\Lambda}_{N-1}$, and $\UN$. Then,
    \begin{align}
        \occupDPVinitx &= \occupDPVinitxZero \nonumber \\
        \ASDPVinitx &= \{\xunpertNoTail\}\oplus\ASDPVinitxZero \nonumber
    \end{align}
\end{prop}
Proposition~\ref{prop:linear_exp} allows the computation of $\ASDPVinitxZero$ be done offline, and an online computation of its Minkowski sum with $\xunpertNoTail$, which is the translation of $\ASDPVinitxZero$ by $\xunpertNoTail$ to obtain $\ASDPVinitx$.
Proposition~\ref{prop:linear_exp} can enable faster planning, especially when solving motion planning problems in an environment with homogeneous rigid body obstacles with stochastic \DPV{} dynamics.


\subsection{Projection-based tight polytopic approximation}
\label{sub:overpoly}

In Section~\ref{sub:tight_polytopic_approx}, we discussed how projections may be used to compute (see Algorithm~\ref{algo:tight}) arbitrarily tight polytopic over/underapproximation of convex and compact sets.
From Theorem~\ref{thm:convexCompactDPV}, we know the sufficient conditions under which $\ASDPVnothing$ is convex and compact.
To compute tight polytopic approximations of $\ASDPVnothing$, we replace the generic projection problem \eqref{prob:projection_problem} with
\begin{mini!}
    { \overline{y}\in \mathcal{X}}{{\Vert \overline{y} - \overline{p}_i \Vert}_2\label{eq:tight_poly_cost}}
{\label{prob:tight_poly}}{}
    \addConstraint{\log(\occupDPVNoTail)}{\geq\log(\alpha)\label{eq:tight_poly_constraint1}}
\end{mini!}
Problem \eqref{prob:tight_poly} is also convex, as guaranteed by Proposition~\ref{prop:OccupDPV_logcon}.
Using Algorithm~\ref{algo:tight}, we obtain two polytopes $\ASDPVNoTailunder$ and $\ASDPVNoTailover$ such that 
\begin{align}
    \ASDPVNoTailunder \subseteq &\ASDPVNoTail\subseteq \ASDPVNoTailover. \nonumber
\end{align}
In addition, we use \eqref{eq:ASDPV_full_empty} to simplify the computation.
We summarize this approach in Algorithm~\ref{algo:ASDPV_poly}.

The implementation of Algorithm~\ref{algo:ASDPV_poly} requires a nonlinear solver like MATLAB's \emph{fmincon} to solve \eqref{prob:tight_poly} due to the nonlinearity of \eqref{eq:tight_poly_constraint1}.
The evaluation of $\occupDPVnothing$ is done via Monte-Carlo simulation using MATLAB's \emph{mvncdf}~\cite{GenzJCGS1992}, resulting in MATLAB's \emph{fmincon} having to minimize under a noisy constraint.
This results in numerical issues like sensitivity to the bounding ball radius $r$ used in Algorithm~\ref{algo:tight} which is not anticipated by the theory (see Section~\ref{sub:tight_polytopic_approx}).
We plan to investigate alternative approaches to solving \eqref{prob:tight_poly} in future.
\begin{algorithm}
    \caption{Projection-based tight polytopic approximations of $\ASDPVNoTail$\label{algo:ASDPV_poly}}
    \begin{algorithmic}[1]    
        \Require{DPV with arbitrary $\bw$, threshold $\alpha\geq 0$, convex and compact rigid body $ \mathcal{O}( \overline{0})$, a desired number of samples $K$ for Algorithm~\ref{algo:tight}, bounding ball radius $r>0$}
        \Ensure{$\ASDPVNoTailunder,\ASDPVNoTailover$}     
        \State $\ymax \gets \max_{ \overline{y}\in \mathcal{X}} \occupDPVNoTail$\Comment{Use Prop.~\ref{prop:ymax} when valid}
        \If{$ \occupDPVymax\leq \alpha$}
            \State $\ASDPVNoTailunder,\ASDPVNoTailover$ $\gets \emptyset$ by \eqref{eq:ASDPV_full_empty}
        \ElsIf{$\alpha> 0$}
        \State Sample  $K$ points $ \overline{p}_i\in \partial\Ball( \ymax, r)$
        \State $\ASDPVNoTailunder,\ASDPVNoTailover$ $\gets$ Algorithm~\ref{algo:tight} with \eqref{prob:tight_poly}
        \Else
            \State $\ASDPVNoTailunder,\ASDPVNoTailover$ $\gets \mathcal{X}$ by \eqref{eq:ASDPV_full_empty}
        \EndIf
  \end{algorithmic}
\end{algorithm}

\subsection{Minkowski sum-based overapproximation}
\label{sub:Mink}

Consider the following set defined using $\mathcal{O}( \overline{0})$ and a superlevel set of $\psibxzNotail$,
\begin{align}
    \ASDPVoverMink=\left\{ \overline{z}\in \mathcal{X}: \psibxzNotail \geq \frac{\alpha}{ \mathrm{m}( \mathcal{O}(\overline{0}))}\right\}\oplus \mathcal{O}( \overline{0}). \label{eq:ASDPVoverMink}
\end{align}
where ${ \mathrm{m}( \mathcal{O}(\overline{0}))}$ refers to the Lebesgue measure of the set $ \mathcal{O}( \overline{0})$.
Proposition~\ref{prop:overapproxMinkSum} ensures that $\ASDPVoverMink$ is an overapproximation of the \avoidsettext{}.
\begin{prop} 
    For any $\alpha\in \mathbb{R}, \alpha\geq 0$, $\tau\in \mathbb{N}_{[1,N]}$ and a bounded FSRPD $\psibxzNotail$, $\ASDPVNoTail\subseteq \ASDPVoverMink$.\label{prop:overapproxMinkSum}
\end{prop}
\begin{proof}
    Let $ \overline{y}\in \ASDPVNoTail \Rightarrow \occupDPVNoTail\geq \alpha$. 
    Recall that for any two non-negative measurable functions $f_1,f_2$ such that $f_1( \overline{z})\leq f_2( \overline{z})$ for all $ \overline{z}\in \mathcal{X}$, we have  $\int_{\mathcal{X}} f_1( \overline{z})d\overline{z}\leq \int_{\mathcal{X}}f_2( \overline{z})d\overline{z}$~\cite[Prop. 19.2.6(c)]{TaoAnalysisII}. 
    By \eqref{eq:occup_DPV_as_prob},
    \begin{align}
        \alpha\leq \occupDPVNoTail&= \int_{ \mathcal{X}} \psibxzNotail 1_{(-\mathcal{O}(-\overline{y}))}( \overline{z}) d\overline{z} \leq \left(\sup_{ \overline{z}\in (-\mathcal{O}(-\overline{y}))}\psibxzNotail\right) \mathrm{m}\left(-\mathcal{O}(-\overline{y})\right).\label{eq:upperboundoccupDPV}  
    \end{align}
If the FSRPD $\psibxzNotail$ is unbounded, the upper bound given in \eqref{eq:upperboundoccupDPV} could be trivially $\infty$.
    From Lemma~\ref{lem:rigidBody}\ref{lem:rigidBodyCenter} and~\cite[Thms. 12.1 and 12.2]{billingsley_probability_1995}, $\mathrm{m}\left(-\mathcal{O}(-\overline{y})\right) = \mathrm{m}\left(\{ \overline{y}\} \oplus (-\mathcal{O}(\overline{0}))\right) = \mathrm{m}\left(\mathcal{O}(\overline{0})\right)$.
    By Lemma~\ref{lem:rigidBody}\ref{lem:rigidBody1},
    \begin{align}
                  \alpha\leq \occupDPVNoTail \Rightarrow \left(\sup_{ \overline{y}\in \mathcal{O}(\overline{z})}\psibxzNotail\right) \mathrm{m}\left(\mathcal{O}(\overline{0})\right) \geq \alpha.\nonumber
    \end{align}
    By Assumption~\ref{assum:Borel}, $ \mathrm{m}( \mathcal{O}(\overline{0}))\neq 0$. 
    Hence, $\exists \overline{z}\in \mathcal{X}$ such that $\overline{y}\in \mathcal{O}(\overline{z})$ and $\psibxzNotail \geq \frac{\alpha}{ \mathrm{m}( \mathcal{O}(\overline{0}))}$.
    From Lemma~\ref{lem:rigidBody}\ref{lem:rigidBodyCenter}, $ \overline{y}\in (\mathcal{O}( \overline{0}) \oplus \{ \overline{z}\})$ where $ \overline{z}$ satisfies the condition $\psibxzNotail \geq \frac{\alpha}{ \mathrm{m}( \mathcal{O}(\overline{0}))}$, which completes the proof.
\end{proof}

Equation \eqref{eq:upperboundoccupDPV} is the only overapproximation step in the proof of Proposition~\ref{prop:overapproxMinkSum}.
This inequality becomes tighter when the ``variation'' of $\psibxzNotail$ in $ - \mathcal{O}( - \overline{y})$ becomes smaller.
For example, when the set $ \mathcal{O}( \overline{0})$ is contained in a relatively small ball.

Recall that support functions may be used to characterize a convex and compact set $ \mathcal{G}\subset \mathcal{X}$~\cite[Thm. 5.6.4]{webster1994convexity}.
We denote the \emph{support function} of $ \mathcal{G}$ by $\rho: \mathcal{X} \rightarrow \mathcal{X}$ 
\begin{align}
    \rho( \overline{l}; \mathcal{G}) = \sup_{ \overline{y}\in \mathcal{G}} \overline{l}^\top \overline{y},\qquad \overline{l}\in \mathbb{R}^n.\label{eq:supp_fn}
\end{align}
From~\cite[Defn 2.1.4]{KurzhanskiTextbook}, we have a closed form expression for the support function of ellipsoid \eqref{eq:ell_defn} with some positive definite $Q\in \mathbb{R}^{n\times n}$, 
\begin{align}
    &\rho\left( \overline{l}; \mathcal{E}\left(\overline{\mu}_{\bx_\tau},Q\right)\right) = \overline{l}^\top \overline{\mu}_{\bx_\tau} + \sqrt{\overline{l}^\top Q\overline{l}}.\label{eq:rho_ell}
\end{align}
For a Gaussian $\bx_\tau$ \eqref{eq:Gauss_psibx_pdf}, \eqref{eq:ASDPVoverMink} simplifies to
\begin{align}
\ASDPVoverMink&= \mathcal{E}\left(\overline{\mu}_{\bx_\tau},Q_{\bx_\tau}\left(\frac{\alpha}{ \mathrm{m}( \mathcal{O}(\overline{0}))}\right)\right) \oplus \mathcal{O}( \overline{0}).\label{eq:ASDPVoverMinkGauss}
\end{align}
The support function of the Minkowski sum of two non-empty, convex, and compact sets is the sum of the respective support functions~\cite[Thm. 5.6.2]{webster1994convexity}.
Therefore, we have a closed-form description of the support function of $\ASDPVoverMink$,
\begin{align}
    \rho( \overline{l}; \ASDPVoverMink) = \overline{l}^\top \overline{\mu}_{\bx_\tau} + \sqrt{\overline{l}^\top Q_{\bx_\tau}\left(\frac{\alpha}{ \mathrm{m}( \mathcal{O}(\overline{0}))}\right)\overline{l}} + \rho( \overline{l}; \mathcal{O}(\overline{0})). \label{eq:support_ASDPVoverMink}
\end{align}
For a polytopic $ \mathcal{O}( \overline{0}) = \{ \overline{y}\in \mathbb{R}^n: H \overline{y}\leq k\}$ with appropriate $H, k$, $\rho( \overline{l}; \mathcal{O}(\overline{0}))$ is solved via a linear program given $ \overline{l}$~\cite[Sec. 4.1]{le_guernic_reachability_2009}.
For an ellipsoidal $ \mathcal{O}( \overline{0})=\mathcal{E}(\overline{0},Q_\mathcal{O})$ with an appropriate $Q_\mathcal{O}\in \mathbb{R}^{n\times n}$, $\rho( \overline{l}; \mathcal{O}( \overline{0}))$ is given by \eqref{eq:rho_ell}.
Alternatively, we can use the ellipsoidal toolbox (\texttt{ET})~\cite{ET} to compute ellipsoidal overapproximations of $\ASDPVoverMink$, as Minkowski sums of ellipsoids need not be ellipsoids~\cite[Pg. 97]{KurzhanskiTextbook}.
Note that $\ASDPVoverMink$ is an extension of the results in~\cite{du2011probabilistic} for arbitrary robot and obstacle rigid body shapes and Gaussian-perturbed obstacle dynamics.

In general, given a support function, a tight polytopic overapproximation of $\ASDPVoverMink$ can be found by ``sampling'' the direction vectors $ \overline{l}$~\cite[Prop. 3]{le_guernic_reachability_2009}.
Specifically, using the support function of $\ASDPVoverMink$ given in \eqref{eq:support_ASDPVoverMink}, we define a tight polytopic overapproximation
\begin{align}
    \OvASDPVoverMink&=\{ \overline{y}\in \mathbb{R}^n: A_\mathrm{des} \overline{y} \leq \overline{b}\} \label{eq:tightASDPVoverMink}
\end{align}
with $N_\mathrm{des}>0$, $A_\mathrm{des}={[ \overline{a}_1^\top\ \overline{a}_2^\top\ \ldots\ \overline{a}_{N_\mathrm{des}}^\top]}^\top$ with $ \overline{a}_{i} \in \mathbb{R}^n$ as the desired supporting hyperplane directions, and $ \overline{b}={[b_1\ b_2\ \ldots\ b_{N_\mathrm{des}}]}^\top\in \mathbb{R}^{N_\mathrm{des}}$ with
\begin{align}
    b_i = \rho( \overline{a}_i; \ASDPVoverMink). &\forall i\in \mathbb{N}_{[1,N_\mathrm{des}]}.\label{eq:supporting_hyp_ASDPVoverMink}
\end{align}
Here, tightness refers to the fact that the supporting hyperplanes of the polytope $\OvASDPVoverMink$, $ \overline{a}_i^\top \overline{y}\leq b_i$, support the set $\ASDPVoverMink$.
Algorithm~\ref{algo:MinkSum} computes $\OvASDPVoverMink$ for a Gaussian $\bx_\tau$ using   $A_\mathrm{des}$ and \eqref{eq:supporting_hyp_ASDPVoverMink}.
Algorithm~\ref{algo:MinkSum} does not require numerical quadrature like Algorithm~\ref{algo:ASDPV_poly}, and works well even when $ \mathcal{O}( \overline{0})$ is contained in a relatively small ball.
\begin{algorithm}
    \caption{Minkowski sum-based approximation of $\ASDPVnothing$ for a Gaussian FSRPD\label{algo:MinkSum}}
    \begin{algorithmic}[1]    
        \Require{DPV with Gaussian disturbance $\bw$, threshold $\alpha\geq 0$, support function of the rigid body $ \rho( \overline{l};\mathcal{O}( \overline{0}))$, matrix of desired supporting hyperplane directions $A_\mathrm{des}$}
        \Ensure{$\OvASDPVoverMink$}     
        \State $\ymax \gets \max_{ \overline{y}\in \mathcal{X}} \occupDPVNoTail$\Comment{Use Prop.~\ref{prop:ymax} when valid}
        \If{$ \occupDPVymax\leq \alpha$}
        \State $\OvASDPVoverMink\gets \emptyset$ by \eqref{eq:ASDPV_full_empty}
        \ElsIf{$\alpha> 0$}
            \For{$i\in \mathbb{N}_{[1,N_\mathrm{des}]}$}
                \State $b_i \gets \rho( \overline{a}_i; \ASDPVoverMink) $ by \eqref{eq:support_ASDPVoverMink} and \eqref{eq:supporting_hyp_ASDPVoverMink}
            \EndFor
            \State $\OvASDPVoverMink\gets \{\overline{y}\in \mathbb{R}^n: A_\mathrm{des} \overline{y} \leq \overline{b}\}$
        \Else
            \State $\OvASDPVoverMink\gets \mathcal{X}$  by \eqref{eq:ASDPV_full_empty}
        \EndIf
  \end{algorithmic}
\end{algorithm}

For a non-Gaussian disturbance $\bw$, the FSRPD may be obtained through Corollary~\ref{corr:FSRPD_def_DPV}.
We can then utilize Algorithm~\ref{algo:tight} to compute a tight polytopic overapproximation of the level set of the FSRPD and, if required, $ \mathcal{O}( \overline{0})$.
In this case, the set $\ASDPVoverMink$ is the Minkowski sum of two polytopes, which can be easily computed using the Multi-Parametric Toolbox (\texttt{MPT3})~\cite{MPT3}.
\begin{rem}
    Proposition~\ref{prop:linear_exp} holds true for the approximations 
    $\ASDPVNoTailunder$, $\ASDPVNoTailover$, $\ASDPVoverMink$, and $\OvASDPVoverMink$.
\end{rem}

Algorithms~\ref{algo:ASDPV_poly} and~\ref{algo:MinkSum} can provide higher quality approximations at the cost of computational time.
For Algorithm~\ref{algo:ASDPV_poly}, using more external points $ \overline{p}_i$ in Algorithm~\ref{algo:tight}  (increasing $K$) based on $\ASDPVNoTailunder$ can yield tighter overapproximations of $\ASDPVNoTail$.
For Algorithm~\ref{algo:MinkSum}, using a larger set of direction vectors $A_\mathrm{des}$ tightens the overapproximation of $\ASDPVoverMink\supset\ASDPVNoTail$.


\section{Numerical simulations}
\label{sec:num}

All computations were performed using MATLAB on an Intel Xeon CPU with 3.4GHz clock rate and 32 GB RAM.
All polyhedral computations were done using \texttt{MPT3}~\cite{MPT3}.
We will consider centrally symmetric $ \mathcal{O}( \overline{0})$ (origin-centered boxes/balls) which enables the use of Proposition~\ref{prop:OccupDPV_symm} for the definition of the occupancy function, and Proposition~\ref{prop:ymax} for avoiding the computation of $\ymax$.

\subsection{Quality of the approximations of $\ASDPVnothing$}

\begin{figure*}
\centering
\newcommand{\figw}{0.24}
\subfloat[]{\includegraphics[width=\figw\linewidth, Trim=\trimValuesNoLegend, clip]{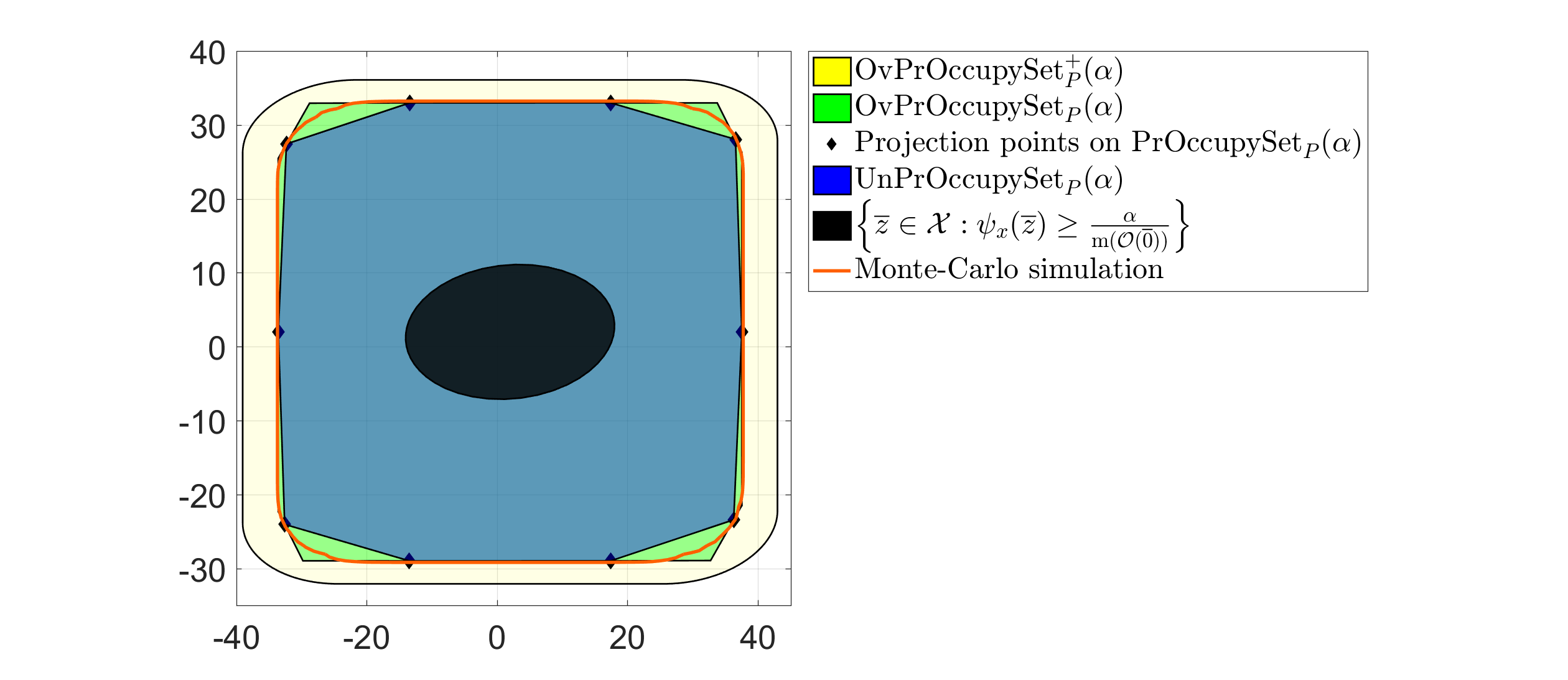}}
\quad
\subfloat[]{\includegraphics[width=\figw\linewidth, Trim=\trimValuesNoLegend, clip]{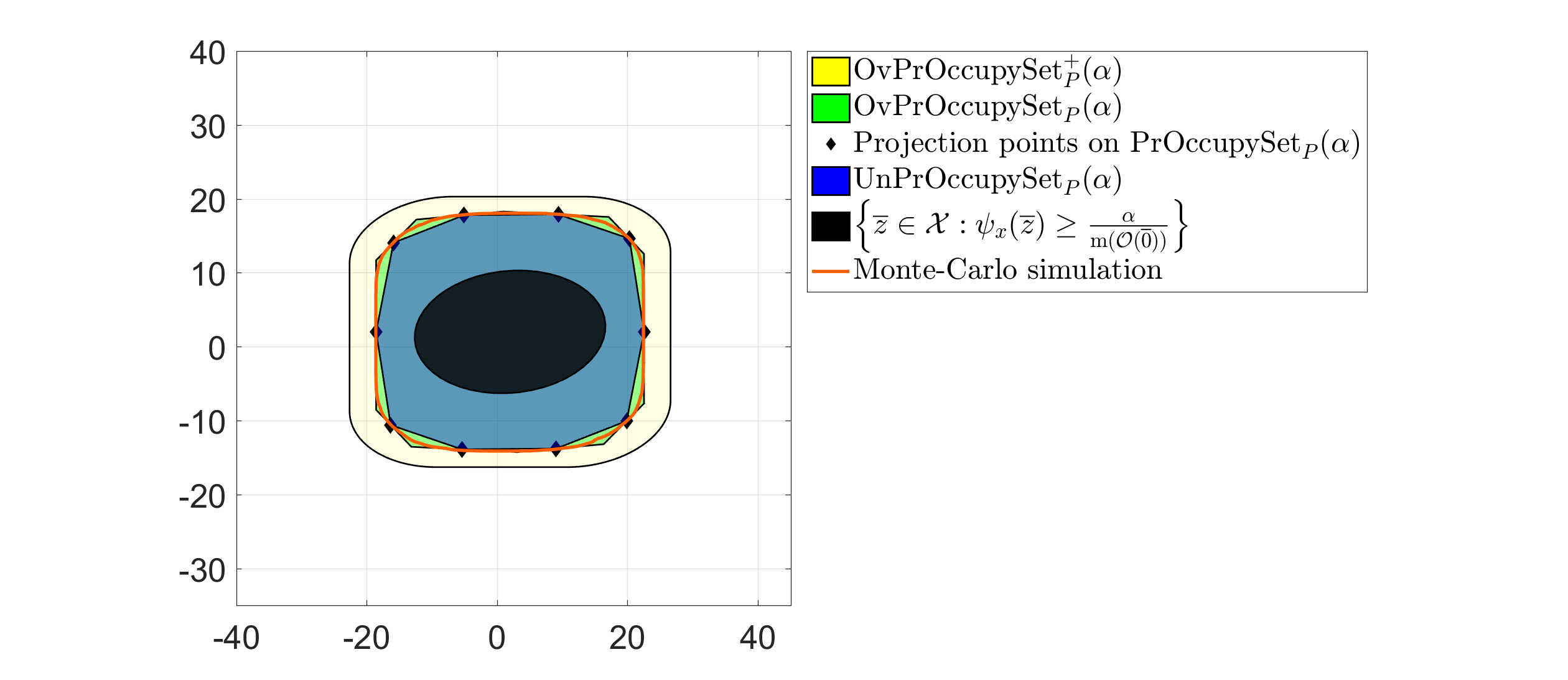}}
\quad
\subfloat[]{\includegraphics[width=0.46\linewidth, Trim=\trimValuesWithLegend, clip]{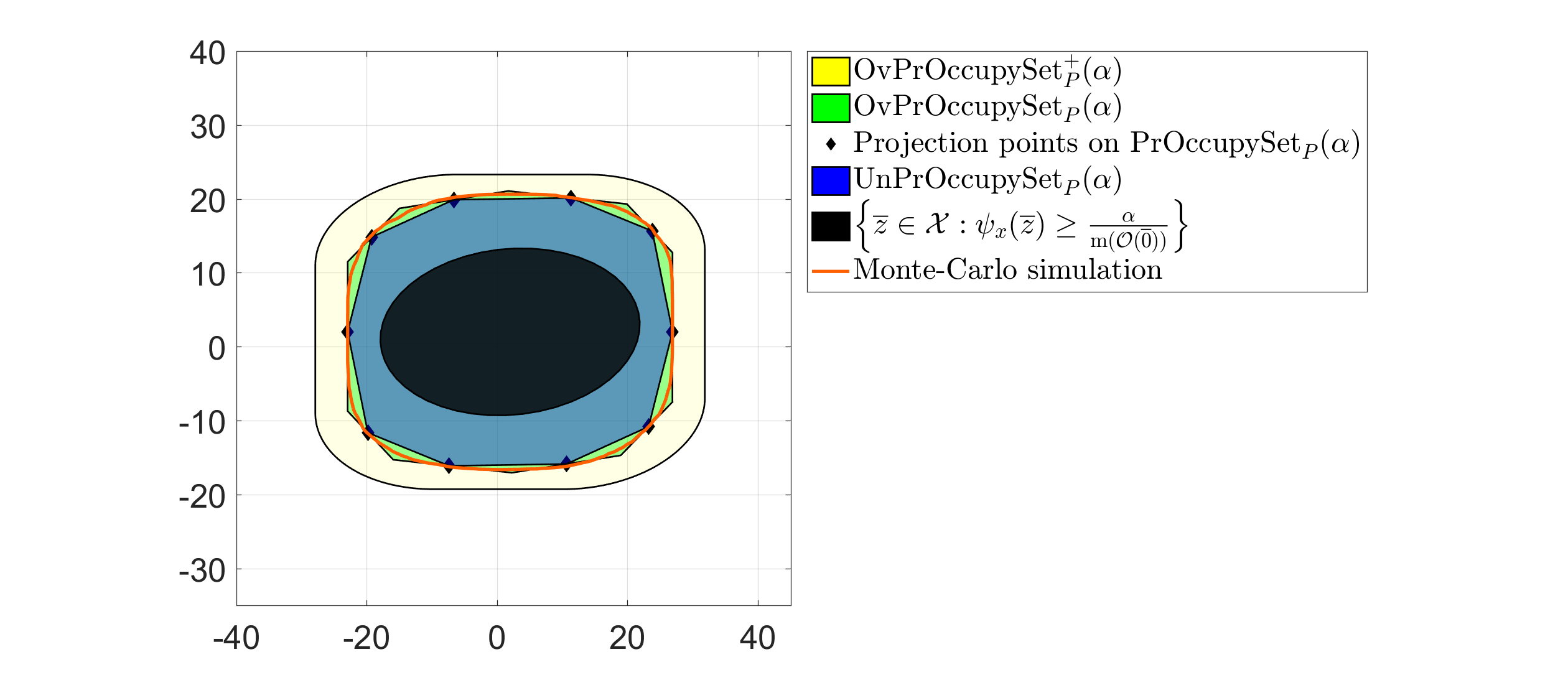}}
\caption{
    Comparison of the tight polytopic approximations $\ASDPVNoTailover$ and $\ASDPVNoTailunder$ (Algorithm~\ref{algo:ASDPV_poly}), the Minkowski sum-based overapproximation $\OvASDPVoverMink$ (Algorithm~\ref{algo:MinkSum}), and the Monte-Carlo simulation-based computation ($10^5$ samples) of the $\ASDPVnothing$ for a Gaussian FSRPD \eqref{eq:Gauss_psibx_pdf} with $\alpha=0.001$. The obstacle shape $ \mathcal{O}( \overline{0})$ is a box of side $50$ in (a) and side $20$ in (b) and (c).
The covariance matrix is same in (a) and (b), and is doubled in (c). 
The computation time is given in Table~\ref{tab:compute_time_VA}.\label{fig:approxCompare}}
\end{figure*}

Figure~\ref{fig:approxCompare} compares the results of Algorithms~\ref{algo:ASDPV_poly} and~\ref{algo:MinkSum} for a Gaussian FSRPD with mean $\mu_{\bx_{\tau}}=[2\ 2]^\top$ under different covariance matrices $\Sigma_{\bx_{\tau}}$ and obstacle geometries $ \mathcal{O}( \overline{0})$ at some $\tau>0$.
Figure~\ref{fig:approxCompare}.a and~\ref{fig:approxCompare}.b uses $\Sigma_{\bx_{\tau}} = H$ and Figure~\ref{fig:approxCompare}.c uses $\Sigma_{\bx_{\tau}} = 2\cdot H$ where $H=\left[ {\begin{array}{cc}    11.62  &    0.59 \\ 0.59 &   3.75   \end{array} } \right]$.
Figure~\ref{fig:approxCompare}.a uses $ \mathcal{O}( \overline{0}) = \Mybox( \overline{0}, 50)$ and Figure~\ref{fig:approxCompare}.b and~\ref{fig:approxCompare}.c uses $ \mathcal{O}( \overline{0}) = \Mybox( \overline{0}, 20)$.
Using $\alpha=0.001$, the \avoidsettext{} is the set of all states that have a probability of being occupied by the rigid body is greater than $0.001$ when its state $\bx_\tau\sim \mathcal{N}(\mu_{\bx_{\tau}},\Sigma_{\bx_{\tau}})$.
Therefore, avoiding this set guarantees probabilistic safety of at least $0.999$.

We used a Monte-Carlo simulation with $10^5$ samples as the ground truth $\ASDPVNoTail$.
The sets provided by Algorithm~\ref{algo:ASDPV_poly}, $\ASDPVNoTailover$ and $\ASDPVNoTailunder$, tightly approximates the set $\ASDPVNoTail$, as expected.
Moreover, $\ASDPVNoTailover$ provides a tighter overapproximation of $\ASDPVNoTail$ than $\OvASDPVoverMink$ obtained from Algorithm~\ref{algo:MinkSum} for a sufficiently high $K$ (here, $K=10$) and well spread-out external points $ \overline{p}_i$ in Algorithm~\ref{algo:ASDPV_poly}.
From Table~\ref{tab:compute_time_VA}, Algorithms~\ref{algo:ASDPV_poly} and~\ref{algo:MinkSum} completed in under a half of a second, while the Monte-Carlo simulation took significantly longer due to the obstacle size and the use of a grid.

\subsection{Obstacle prediction for rigid body obstacles with unicycle dynamics}

We consider the unicycle dynamics \eqref{eq:obs_DMSP} with parameters described in Appendix~\ref{app:UnicycleVals}.
Figure~\ref{fig:unicycle} shows the \avoidsettext{} with $\alpha=0.01$ after $15$ time steps. 
Using Theorem~\ref{thm:AvoidSetUnionDMSP}, we can compute a cover for the \avoidsettext{} consisting of convex and compact sets.
We present only the results from Algorithm~\ref{algo:MinkSum} due to Algorithm~\ref{algo:ASDPV_poly}'s numerical instability for $ \mathcal{O}( \overline{0})=\Ball( \overline{0},0.2)$.
Figure~\ref{fig:unicycle}.a shows that the ``banana distribution'' seen in the simultaneous localization and mapping literature~\cite{long_banana_2013} is covered by the proposed overapproximation.
The modified threshold $\alphaSqtau$ eliminates the computation of covers for empty sets leading to a minor decrease in computation time in Figure~\ref{fig:unicycle}.b.

As seen in Figure~\ref{fig:unicycle_FSRPD}, the \DPV{} correponding to $\bom=0$ (denoted by $ \overline{\Lambda}_\tau'$, the black $\diamond$s in Figure~\ref{fig:unicycle}.b) forces the $\FSRsetDPVNoTail$ to be the line corresponding to the column space of $ \mathscr{C}_W(\tau;\cdot,\overline{\Lambda}_\tau')$.
In this case, the point $ \overline{y}$ is occupied when $\bx_\tau$ lies in the line segment corresponding to the intersection of $ \mathcal{O}( \overline{0})$ with $\FSRsetDPVNoTail$ (since $\bx_\tau$ can not lie outside $\FSRsetDPVNoTail$). 
By Proposition~\ref{prop:OccupDPV_symm}, the occupancy function simplifies for this \DPV{} to an integral of the projected one-dimensional Gaussian density over an interval.
Similar projections may be used to deal with arbitrary \DPV{} dynamics with $ \mathscr{C}_W(\tau;\cdot)$ that does not have full rank.

For obstacle avoidance problems, the collection of sets shown in Figure~\ref{fig:unicycle} is the set of deterministic keep-out regions to achieve a desired collision-avoidance probability (less than $0.01$ in this case).
We will explore the efficacy of using these sets in obstacle avoidance in future work.

\begin{figure}
    \centering
    \newcommand{\figw}{0.28}
    \subfloat[]{\includegraphics[width=\figw\linewidth,Trim=\trimValues,clip]{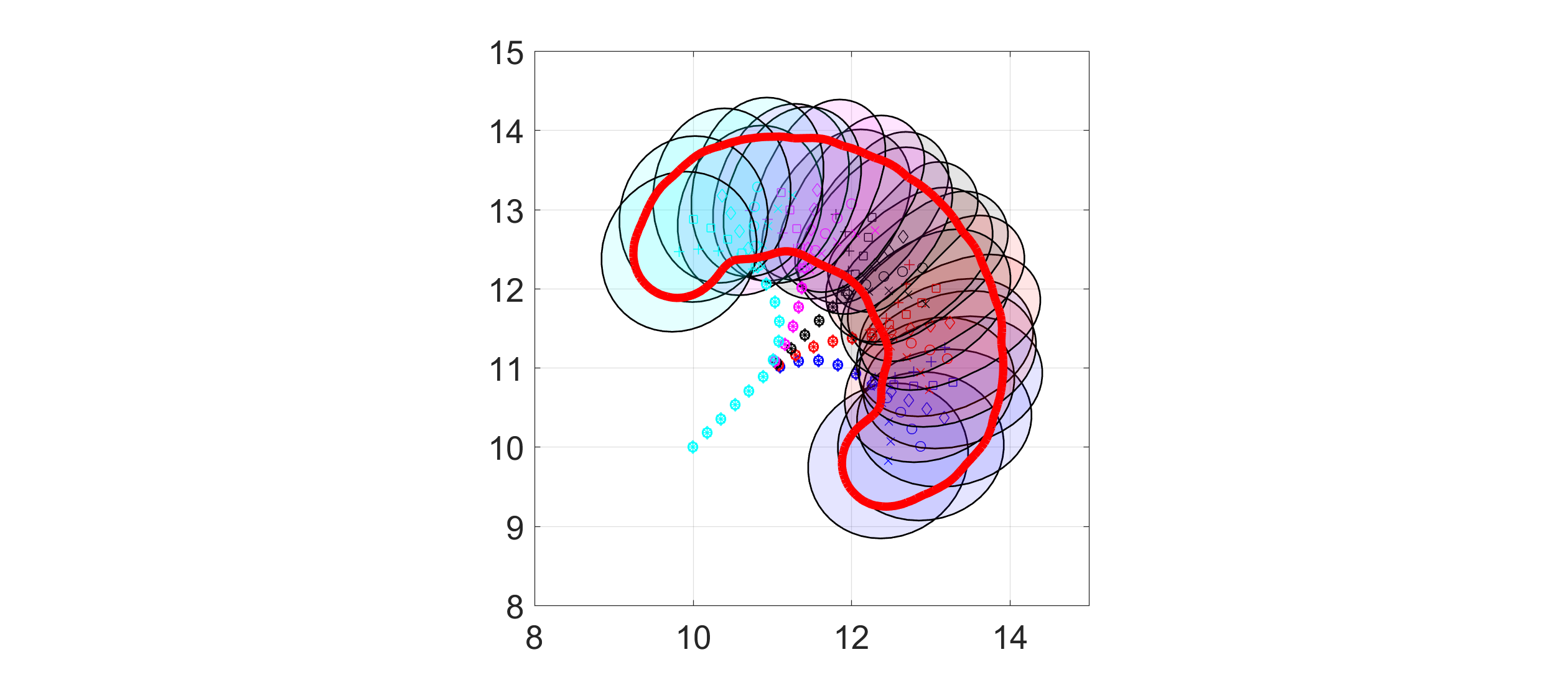}} 
    \ \subfloat[]{\includegraphics[width=\figw\linewidth,Trim=\trimValues,clip]{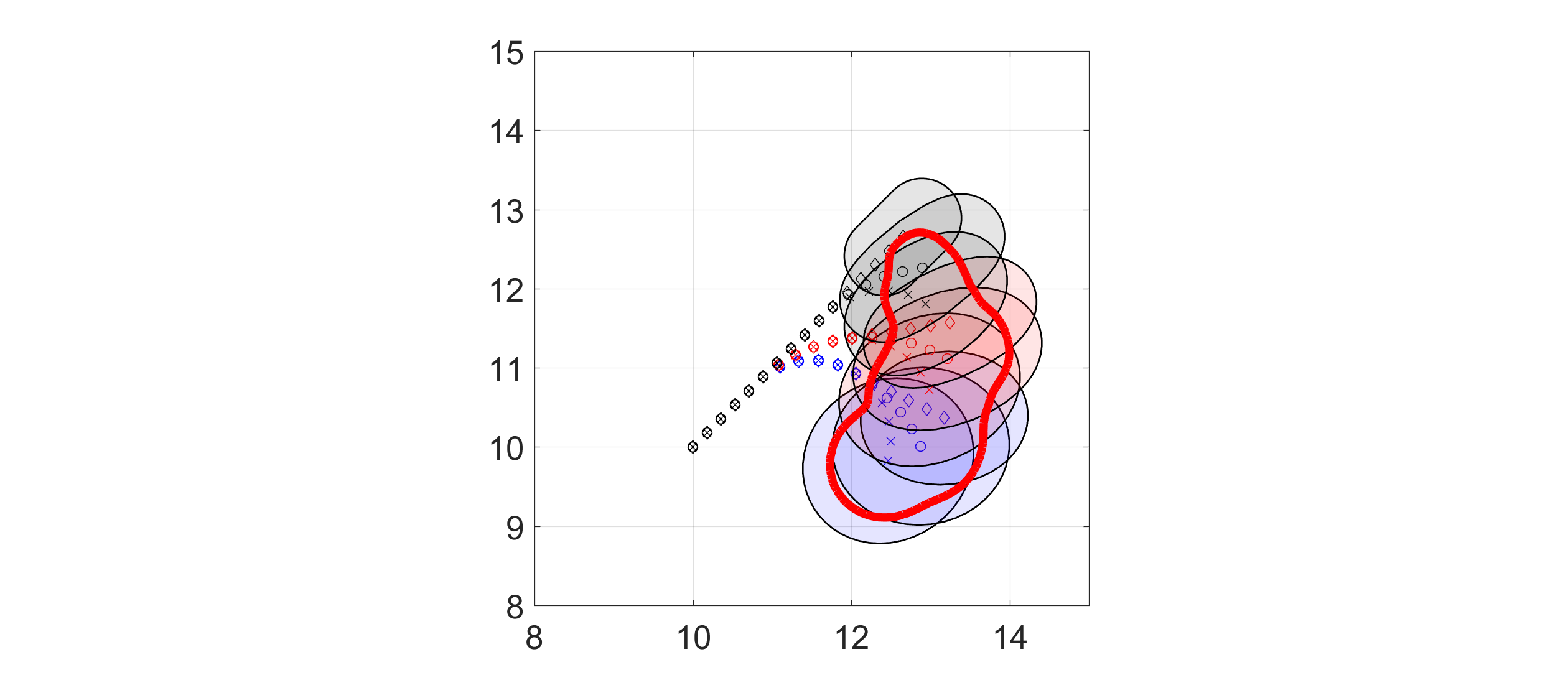}}
    \caption{Overapproximation of the $\ASDMSPnothing{}$ (red contour via Monte-Carlo simulation, $10^5$ samples) for unicycle dynamics \eqref{eq:obs_DMSP} using Algorithm~\ref{algo:MinkSum} and Theorem~\ref{thm:AvoidSetUnionDMSP} (see Appendix~\ref{app:UnicycleVals} for numerical values). Figure~\ref{fig:unicycle}a is the ``banana'' distribution~\cite{long_banana_2013}. The markers indicate the mean state trajectories of the respective \DPV{} dynamics, and match in color to their respective $\ASDPVnothing$. The computation times are given in Table~\ref{tab:compute_time_VA}.}
    \label{fig:unicycle}
\end{figure}

\begin{table}
    \centering
    \begin{tabular}{|c|c|c|c|}
    \hline									
    Figures & Algorithm~\ref{algo:ASDPV_poly} & Algorithm~\ref{algo:MinkSum} & Monte-Carlo approach \\\hline
    Figure~\ref{fig:approxCompare}.a & $0.30$ s & $0.25$ s & $95.89$ s\\ \hline
    Figure~\ref{fig:approxCompare}.b & $0.32$ s & $0.25$ s & $28.44$ s\\ \hline
    Figure~\ref{fig:approxCompare}.c & $0.42$ s & $0.24$ s & $22.77$ s \\ \hline
    Figure~\ref{fig:unicycle}.a & -- & $0.15$ s & $18.67$ s \\ \hline
    Figure~\ref{fig:unicycle}.b & -- & $0.09$ s & $16.10$ s \\ \hline
   \end{tabular}
   \caption{Computation time for Figures~\ref{fig:approxCompare} and~\ref{fig:unicycle} using Algorithms~\ref{algo:ASDPV_poly} and~\ref{algo:MinkSum}, and Monte-Carlo simulation ($10^5$ samples).}
   \label{tab:compute_time_VA}
\end{table}

\section{Conclusion and future work}
\label{sec:conc}

In this paper, we have analyzed the \occupfuntext{} and the \avoidsettext{} of rigid body obstacles with \DPV{}/\DMSP{} dynamics using forward stochastic reachability.
We developed forward stochastic reachability techniques for these dynamics and analyzed their convexity.
Additionally, we characterized sufficient conditions for the log-concavity and the upper semi-continuity of the \occupfuntext{}.
Using these results, we characterized sufficient conditions for the convexity, closedness, and compactness of the \avoidsettext{}.
We proposed two computationally efficient algorithms to compute the approximation of the \avoidsettext{}.

In future, we would like to use the computed \avoidsettext{} to generate probabilistically safe dynamically-feasible trajectories for robots navigating an environment with rigid body obstacles with \DPV{}/\DMSP{} dynamics.
Additionally, we would also like to investigate alternative ways to exploit the convexity of the problem \eqref{prob:tight_poly} to propose an approach to compute the projection point without using a generic nonlinear solver like MATLAB's \emph{fmincon}.

\section*{Acknowledgments}

We would like to thank Baisravan HomChaudhuri for his insightful comments on this research.

\appendix

\subsection{Parameters for the unicycle example \eqref{eq:obs_DMSP}}
\label{app:UnicycleVals}

We consider a unicycle dynamics \eqref{eq:obs_nonlin_sys} modeled as a \DMSP{} \eqref{eq:obs_DMSP} with sampling time $T_s = 0.05$, initial turning rate $ \omega_0 = 0$, initial heading (parameter value) $\lambda_0 = \frac{\pi}{4}$, and Gaussian velocity random vector $\bv \sim \mathcal{N}(5,1)$. 
The set of admissible turning rates is $ \mathcal{Q}= \{-5,-2.5,0,2.5,5\}$.

For Figure~\ref{fig:unicycle_FSRPD}$, $ we define 
\begin{align}
\overline{q}_{\tau-1} = [0, 0, \underbrace{-5, \ldots, -5}_{\mbox{$19$ times}}, \underbrace{0, \ldots, 0}_{\mbox{$6$ times}}, \underbrace{5, \ldots, 5}_{\mbox{$18$ times}}] \nonumber
\end{align}
and the initial obstacle location is $\overline{x}_0$ is at the origin.

For Figure~\ref{fig:unicycle}, the unicycle switches its turning rate based on the switching law \eqref{eq:MarkovSwitchLaw} after every $\tausw = 5$ time steps using a transition matrix $M\in \mathbb{R}^{5\times 5}$, and the initial obstacle location is $\overline{x}_0={[10\ 10]}^\top$.
Figure~\ref{fig:unicycle}.a uses $M_1$ (chose with equal likelihood among the turning rates), and Figure~\ref{fig:unicycle}.b uses $M_2$ (prefer taking right over staying straight and never take left),
\begin{align}
    \small
    M_1= \left[ {\arraycolsep=2pt\begin{array}{ccccc}
                      0.2 & 0.2 & 0.2 & 0.2 & 0.2 \\   
                      0.2 & 0.2 & 0.2 & 0.2 & 0.2 \\   
                      0.2 & 0.2 & 0.2 & 0.2 & 0.2 \\   
                      0.2 & 0.2 & 0.2 & 0.2 & 0.2 \\   
                      0.2 & 0.2 & 0.2 & 0.2 & 0.2 \\   
                  \end{array} } \right],\ 
    M_2= \left[ {\arraycolsep=2pt\begin{array}{ccccc}
                      0.5 & 0.47 & 0.03 & 0 & 0 \\   
                      0.5 & 0.47 & 0.03 & 0 & 0 \\   
                      0.5 & 0.47 & 0.03 & 0 & 0 \\   
                      0.5 & 0.47 & 0.03 & 0 & 0 \\   
                      0.5 & 0.47 & 0.03 & 0 & 0 \\   
                  \end{array} } \right].\nonumber 
\end{align}

\subsection{Proof of Lemma~\ref{lem:rigidBody}}
\label{app:proof_lem_rigid_body}
    To show~\ref{lem:rigidBodyCenter}), define $ \overline{z}'= \overline{z}- \overline{y}$.
    From \eqref{eq:obs_rigidbody_defn}, we have $ \mathcal{O}( \overline{y}) = \{\overline{z}'+ \overline{y}\in \mathcal{X}: h(\overline{z}')\geq 0\} = \{ \overline{y}\} \oplus \{ \overline{z}: h(z) \geq 0\}=\{ \overline{y}\} \oplus \mathcal{O}( \overline{0})$.

    Next, we show \ref{lem:rigidBody1}) and~\ref{lem:rigidBody2}). Consider $ \overline{y},\overline{z}\in \mathcal{X}$ such that $\overline{z}\in(-\mathcal{O}(-\overline{y}))$.
    \begin{align}
        \overline{z}\in(-\mathcal{O}(-\overline{y}))\Longleftrightarrow -\overline{z}\in \mathcal{O}(-\overline{y})\Longleftrightarrow \hobs(-\overline{z}-(-\overline{y}))\geq 0 &\Longleftrightarrow \hobs(\overline{y}-\overline{z})\geq 0 \label{eq:zminusO}\\
                                                    &\Longleftrightarrow \overline{y}\in \mathcal{O}( \overline{z}) \label{eq:part1}\\
                                                    &\Longleftrightarrow \overline{y}\in \{\overline{z}\} \oplus \mathcal{O}( \overline{0}) \nonumber \\
                                                    &\Longleftrightarrow (\overline{y}- \overline{z}) \in \mathcal{O}( \overline{0}).\label{eq:part2a} 
    \end{align}
    Equation \eqref{eq:part1} shows~\ref{lem:rigidBody1}), and \eqref{eq:part2a} shows~\ref{lem:rigidBody2}).

\subsection{Proof of Lemma~\ref{lem:algo_proof}}
\label{app:proof_algo_proof}
    \emph{Approximation:} The optimization problem \eqref{prob:projection_problem} has a unique optimal solution since $ \mathcal{L}$ is convex, closed, and non-empty~\cite[Thm. 2.4.1]{webster1994convexity}. 
The hyperplane in \eqref{eq:hyperplane_defn} is the supporting hyperplane of $ \mathcal{L}$ at $ \ProjL( \overline{p}_i)$~\cite[Sec. 2.5.2]{boyd_convex_2004}~\cite[Thm. 2.4.1]{webster1994convexity}.
Note that the set of $\ProjL(\overline{p}_i)$ is a subset of the extreme points of $ \mathcal{L}$ and the set of hyperplanes defined using \eqref{eq:hyperplane_defn} is a subset of all the closed halfspaces containing $ \mathcal{L}$.
Hence, we have $ \Linner(K) \subseteq \mathcal{L} \subseteq \Louter(K)$~\cite[Thm 2.6.16, Corr. 2.4.8]{webster1994convexity}.

\emph{Tightness:} Increasing the number of external points $ \overline{p}_i$ to $K^+>K$ (while retaining the previously used external points), we have by the same arguments~\cite[Thm 2.6.16, Corr. 2.4.8]{webster1994convexity} as above,
\begin{align}
    \Linner(K) \subseteq\Linner(K^+) \subseteq \mathcal{L} \subseteq \Louter(K^+)\subseteq\Louter(K).\nonumber
\end{align}
We thus have a monotone increasing sequence of polytopes in $\Linner$ and a monotone decreasing sequence of polytopes in $ \Louter$ with increasing $K$~\cite[Sec. 1]{ChowProbability1997}.
Therefore, 
\begin{align}
    \lim_{K \rightarrow \infty}\Linner(K) &= \cup_{K=1}^\infty\Linner(K)= \mathcal{L}, &\mbox{ and} \nonumber \\
    \lim_{K \rightarrow \infty}\Louter(K) &= \cap_{K=1}^\infty\Louter(K)= \mathcal{L}. \nonumber
\end{align}

\subsection{Proof of Lemma~\ref{lem:symmetric_ProbbWtau}}
\label{app:proof_symmetric_ProbbWtau}

    From the IID assumption, $\ProbOvBwtauNoTail$ is centrally symmetric.
    For any $ \mathcal{G}\in\sigmaAlg( \mathcal{W}^{\tau-1})$,
   \begin{align}
       \ProbbWtau\{\bW_\tau\in -\mathcal{G}\} &=\ProbOvBwtauNoTail\{\mathscr{C}_W(\tau;\cdot) \overline{\bw}_{\tau-1} \in -\mathcal{G}\} \nonumber \\
                         &=\ProbOvBwtauNoTail\{\overline{\bw}_{\tau-1} \in \{ \overline{y}\in \mathcal{W}^{\tau-1}: \mathscr{C}_W(\tau;\cdot)\overline{y}\in -\mathcal{G}\}\} \nonumber \\
            &=\ProbOvBwtauNoTail\{\overline{\bw}_{\tau-1} \in -\{ \overline{y}\in \mathcal{W}^{\tau-1}: \mathscr{C}_W(\tau;\cdot)\overline{y}\in \mathcal{G}\}\} \nonumber \\
            &=\ProbOvBwtauNoTail\{\overline{\bw}_{\tau-1} \in \{ \overline{y}\in \mathcal{W}^{\tau-1}: \mathscr{C}_W(\tau;\cdot)\overline{y}\in \mathcal{G}\}\} \nonumber \\
            &=\ProbbWtau\{\bW_\tau\in \mathcal{G}\}. \nonumber
   \end{align} 

\bibliographystyle{IEEEtran}
\bibliography{IEEEabrv,shortIEEE,FSRonSLS}

\end{document}